\documentclass{article}

\usepackage{arxiv}

\usepackage[utf8]{inputenc}
\usepackage[T1]{fontenc}
\usepackage{hyperref}
\usepackage{url}
\usepackage{booktabs}
\usepackage{amsfonts}
\usepackage{nicefrac}
\usepackage{microtype}
\usepackage{lipsum}
\usepackage{graphicx}
\usepackage{subfig}
\usepackage{bbm}
\usepackage{xcolor}
\usepackage{algorithm}
\usepackage{algpseudocode}
\usepackage{rotating}
\usepackage{pdflscape}
\usepackage{float}

\usepackage{amsmath}
\usepackage{amssymb}
\usepackage{amsthm}

\newtheorem{definition}{Definition}
\newtheorem{theorem}{Theorem}
\newtheorem{remark}{Remark}
\newtheorem{proposition}{Proposition}
\newtheorem{assumption}{Assumption}
\newtheorem{corollary}{Corollary}
\newtheorem{lemma}{Lemma}

\graphicspath{ {./images/} }

\title{Optimized Multi-Level Monte Carlo Parametrization and Antithetic Sampling for Nested Simulations}

\author{
 Alexandre Boumezoued\thanks{Milliman R\&D, 14 avenue de la Grande Armée, Paris, France}
  \And
 Adel Cherchali\footnotemark[1]
  \And
  Vincent Lemaire\thanks{Laboratoire de Probabilités, Statistique et Modélisation (LPSM), 4 rue Jussieu, France}
  \And
  Gilles Pagès\footnotemark[2]
  \And
  Mathieu Truc\thanks{Laboratoire de Probabilités, Statistique et Modélisation (LPSM), Milliman R\&D, 14 avenue de la Grande Armée, Paris, France, email : \texttt{mathieu.truc@milliman.com}}
}

\begin{document}
\maketitle
\begin{abstract}
Estimating risk measures such as large loss probabilities and Value-At-Risk is fundamental in financial risk management and often relies on computationally intensive nested Monte Carlo methods. While Multi-Level Monte Carlo (MLMC) techniques and their weighted variants are typically more efficient, their effectiveness tends to deteriorate when dealing with irregular functions, notably indicator functions, which are intrinsic to these risk measures. We address this issue by introducing a novel MLMC parametrization that significantly improves performance in practical, non-asymptotic settings while maintaining theoretical asymptotic guarantees. We also prove that antithetic sampling of MLMC levels enhances efficiency regardless of the regularity of the underlying function. Numerical experiments motivated by the calculation of economic capital in a life insurance context confirm the practical value of our approach for estimating loss probabilities and quantiles, bridging theoretical advances and practical requirements in financial risk estimation.
\end{abstract}

\section{Introduction}

Let $(\Omega, \mathcal{F}, \mathbb{P})$ be a probability space on which all considered random variables will be defined in this article. Letting $d$ and $q$ be positive integers, we consider two random variables $X$ valued in $\mathbb{R}^d$ and $U$ valued in $\mathbb{R}^q$ independent from $X$. Let $F : \mathbb{R}^{d} \times \mathbb{R}^q \longrightarrow \mathbb{R}$ be a Borel measurable function and denote $Z := F(X, U)$. We make the general assumption that the distribution of $X$ is absolutely continuous (with respect to the Lebesgue measure) and that $Z$ is square integrable. Let $\psi : x \longmapsto \mathbb{E}[F(x, U)]$ be the function from $\mathbb{R}^d$ to $\mathbb{R}$ then $\psi$ is a regular version of the conditional expectation of $Z$ given $X$. Finally we define $L := \psi(X)$ and consider a measurable function $f : \mathbb{R} \longrightarrow \mathbb{R}$. Note that $\mathbb{P}$ almost surely (denoted a.s) we have $L = \mathbb{E}[Z | X]$.

In many practical applications of financial risk management, we need to estimate, by Monte Carlo 
(MC) simulations, quantities of the form 
\begin{equation}
\label{eq:I_def}
    I := \mathbb{E}[f(L)],
\end{equation}
sometimes also quantiles of $L$ are of interest. Most of the time in practice, the conditional expectation $L$ cannot be sampled exactly. 
In that case, one must rely on using approximate samples of $L$ which inexorably introduce a bias in the estimation. A general way to construct random variables approximating $L$ is to consider its MC approximation.

\begin{definition}
    Let $K \in \mathbb{N}$, $x \in \mathbb{R}^d$ and $(U_k)_{k \in \{1,\dots,K\}}$ a sequence of independent and identically distributed  (i.i.d.) copies of $U$ independent from $X$. Then the random variable 
    \begin{equation*}
        \hat{E}_{K}(x) := \frac{1}{K} \sum_{k = 1}^K F(x, U_k)
    \end{equation*}
    is a \emph{MC approximation} of $\psi(x)$ with $K$ inner samples. In addition, the random variable $\hat{E}_K(X)$ is a \emph{MC approximation} of $L$.
\end{definition}
An application of the strong law of large numbers for martingale (\cite{Hall1981-sk}, Theorem 2.18) to the martingale $(M_K)_{K \in \mathbb{N}}$ defined for all $K \in \mathbb{N}$ by $M_K := \sum_{k = 1}^K F(X, U_k) - \mathbb{E}[F(X, U) | X]$ (with respect to the filtration generated by $(X, U_1, \dots, U_K)_{K \in \mathbb{N}}$) ensures that $\hat{E}_K(X)$ converges a.s to $L$ as the number of inner samples $K$ goes to infinity. Using the MC approximation $\hat{E}_K(X)$ in place of $L$ in a MC estimation of $I$ is known as a nested MC framework as it involves two stages of MC approximations. Typically, one would use a nested MC estimator of (\ref{eq:I_def}) defined for $(J, K) \in \mathbb{N}^2$ as 
\begin{equation}
\label{eq:std_nested_mc}
I_{J, K} :=\frac{1}{J} \sum_{j = 1}^J f \left( \frac{1}{K} \sum_{k = 1}^K F(X_k, U_{j, k}) \right) \,,
\end{equation}
where $(X_k)_{k \in \{1, \dots, K\}}$ are i.i.d copies of $X$ and $(U_{j, k})_{(j, k) \in \{1,\dots J\} \times \{1, \dots K\}}$ are i.i.d copies of $U$ independent from $(X_k)_{k \in \{1, \dots, K\}}$. Under mild conditions, this estimator converges in a proper way toward $I$ as $J$ and $K$ goes to infinity (see, e.g., \cite{Rainforth2017OnNM}). Borrowing from the nested MC literature in risk management (e.g., Gordy-Juneja \cite{gordyJuneja10}, Giles \& Haji-Ali \cite{gilesAli19}, Lemaire-Pagès \cite{lemairePages17}), we refer to $\hat{E}_K(X)$  as the inner expectation, the random variables $F(X, U_1), \dots,F(X,U_K)$ as the inner samples (or scenarios) and realizations of $X$ as outer samples (or scenarios).

To alleviate the notations, for all integer $K$, we define $Y_K := f(\hat{E}_K(X))$, the MC approximation of $Y := f(L)$ with $K$ inner samples. For a large class of functions $f$ the weak error $I - \mathbb{E}[Y_K]$ vanishes as $K$ goes to infinity. For example, Gordy-Juneja \cite{gordyJuneja10} (Proposition 1) provide a first order expansion of the weak error when $f$ is an indicator function. Giorgi et. al. \cite{Giorgi_2020} (Proposition 4.1 and Proposition 5.1), proved a higher order expansion when $f$ is a smooth function as well as when $f$ is an indicator function. 

The case when $f$ is an indicator function is particularly relevant in financial risk management.  Indeed if, for $u \in \mathbb{R}$, $x \in \mathbb{R}$, $f(x) = \mathbbm{1}_{x \leq u}$ then $I = \mathbb{P}(L \leq u) = F_L(u)$ is the cumulative distribution function (c.d.f.) of $L$. Then if $L$ is represents a future portfolio loss, we are estimating the probability that the portfolio will incur a loss lower than $u$. It was studied by Gordy-Juneja in their seminal paper \cite{gordyJuneja10}, and later on by many others (e.g., \cite{broadiemoallemi2011}, \cite{gilesAli19}, \cite{Giorgi_2020}).

In practical applications, the sampling of $J \in \mathbb{N}$ i.i.d copies of $Y_K$ typically proceeds in two stages. Initially, $J$ outer scenarios $(X_{1}, \dots, X_{J})$ are drawn independently from the distribution of $X$. Subsequently, $J \times K$ i.i.d. samples $ (U_{j,1}, \dots, U_{j, K}), \; j \in \{1, \dots, J\}$ are drawn from the distribution of $ U $, which are then used to construct the $J \times K$ inner scenarios $F(X_{j}, U_{j, 1}), \dots, F(X_j, U_{j, K}), \; j \in \{1, \dots, J\}$. These are then used to construct $J$ samples $Y^{j}_{K} = f(\frac{1}{K} \sum_{k = 1}^K F(X_j, U_{j, k}))$ of $Y_K$. The unit outer sampling cost is, therefore, defined as the cost of sampling a single $X$, whereas the unit inner sampling cost is defined as the combined cost of sampling one $U$ and evaluating $F(X, U)$. In this paper, we consequently adopt the following computational cost framework for sampling $Y_K$.
\begin{definition}
    For any integer $K$ we define the computational cost of sampling $Y_K$ as
    \begin{equation}
    \label{eq:computational_cost_YK}
        \gamma_{\tau}(K) = \bar{\kappa}(\tau + K) \,,
    \end{equation}
    where $\bar{\kappa} \in (0, +\infty)$ denotes the unit inner sampling cost, and $\tau \in [0, +\infty)$ is a dimensionless constant representing the unit outer sampling cost expressed in units of $\bar{\kappa}$. In other words the cost of sampling $X$ is $\tau$ times the combined cost of sampling $U$ and evaluating $F$.
\end{definition}
For simplicity, and without loss of generality, we set $\bar{\kappa} = 1$, which means all computational costs are measured relative to the unit inner sampling cost. In the context of MLMC the literature typically neglects the cost of sampling $X$ (i.e. assumes $\tau = 0$), however in practice $\tau$ can be large and we show in this paper that this may affect the optimal balance between inner and outer scenarios in non-asymptotic regimes. As an example in a financial risk management context, each sampling of $X$ may require a calibration of multiple Risk-Neutral models with respect to the drawn risk factors. When complex models are considered this calibration phase can be time consuming and requires to consider larger $\tau$ (see \cite{Arrouy2022EconomicSG}).\\

In financial risk management, a common risk measure to compute is the Value-At-Risk
(VaR, see, e.g., \cite{Jorion2000ValueAR}) with high confidence level $\alpha \in (0,1)$ :
\begin{equation}
\label{eq:quantile_definition}
    q_L(\alpha) = \inf \left \{ u \in \mathbb{R} \; | \; F_L(u) \geq \alpha   \right \}.
\end{equation}

As a key example, the 1-year Value-At-Risk with $\alpha=99.5\%$ on future own-fund losses serves as a fundamental measure utilized by the European insurance sector to evaluate the financial situation of companies according to Solvency II regulations (see, e.g., \cite{solvencyBolviken2017}). A closely related risk measure is the evaluation of $F_L(u)$ the c.d.f. of the loss $L$ for extreme threshold $u \in \mathbb{R}$ (i.e $u$ close to extreme quantiles of $L$).

These two measures are directly linked: for continuous distributions, the VaR at level $\alpha$ is the smallest $u$ such that $F_L(u) = \alpha$. Thus, evaluating $F_L(u)$ for a given $u$ gives the probability of losses not exceeding $u$, while VaR inverts this relationship to find the threshold corresponding to a given probability level. The standard approach for practitioners to compute these risk measures is to use a nested Monte Carlo estimator : fix an inner sample size $K$, generate an i.i.d. sample of $\hat{E}_K(X)$ of size $J$, then compute the relevant probability/order statistic (see, e.g., \cite{Bauer2012OnTC}, \cite{NestedStoSoa}). However, nested simulations is known as sample-inefficient when precise estimations are desired. Gordy-Juneja \cite{gordyJuneja10} demonstrate that, to halve the estimation error, the simulation budget must be increased by a factor eight. By comparison, within a conventional unbiased Monte Carlo situation, standard methods require only a fourfold increase in simulation budget to achieve an equivalent reduction in error. 

Therefore, practitioners often resort to alternative techniques that involve statistical learning of the conditional expectation function $\psi$ (also known as proxy methods). Common approaches include Least-Squares Monte Carlo methods, which assumes that $\psi$ can be well-approximated by a polynomial function (see, e.g., \cite{Krah2018ALM}, \cite{Bauer22}), and Replicating Portfolio methods, where $\psi$ is approximated by a portfolio of vanilla options (see, e.g., \cite{cambou2017RP}). Nonetheless, we can identify two main shortcomings with these proxy methodologies. First, practical implementation often requires making arbitrary choices regarding the model specification, which can impede the accurate approximation of $\psi$. Second, these methods may suffer from inefficiency in high-dimensional contexts, making them difficult to scale as the complexity of risk modeling increases. For instance, insurance company can have up to hundreds of risk factors. The scalability of these methods in high dimension is an active research topic where neural networks play an important role (see, e.g., \cite{Castellani2021MachineLT}, \cite{Cheridito2020AssessingAR}, \cite{Krah2020LeastSquaresMC} or \cite{Perla2025TransformersbasedLS}). Balancing the allocation of inner and outer scenarios as a function of $\tau$ is also a pertinent consideration in proxy methods. Alfonsi et al. \cite{Alfonsi2022HowMI} demonstrate that, as $\tau$ gets large, computational efficiency can be achieved by increasing the number of inner samples per outer samples in the construction of the Least-Squares Monte Carlo proxy.

A common fully Monte-Carlo approach to improving the estimation of $I$ in biased Monte Carlo frameworks is to employ Multi-level Monte Carlo (MLMC) methods (see \cite{giles15}). In his pioneering work \cite{gilesMLMC2008}, Giles introduced the MLMC framework to address bias arising from the discretization of stochastic differential equations. Subsequently, numerous authors have extended this methodology (e.g., \cite{lemairePages17}, \cite{gilesAli19}, \cite{Krumscheid2018MultilevelMC}, \cite{Crepey2024AdaptiveMS}), including applications in the context of nested Monte Carlo. In the most favorable scenario (specifically, when the function $f$ in (\ref{eq:I_def}) is sufficiently smooth) the MLMC method achieves sample efficiency comparable to what an unbiased Monte Carlo method would achieve. However when $I$ is the evaluation of a c.d.f.,
\begin{equation}
\label{eq:proba_large_loss}
    I = \mathbb{E}[\mathbbm{1}_{L \leq u}] = F_L(u)
\end{equation}
for some fixed $u \in \mathbb{R}$, we need to deal with an indicator function in (\ref{eq:I_def}) and its irregularity then hinder the efficiency of MLMC methods.

In this paper, we aim at improving the MLMC method for those cases where $f$ is an indicator function, such as when estimating (\ref{eq:proba_large_loss}). In particular we will build upon the Multi-level Richardson Romberg (ML2R) method (see Lemaire-Pagès \cite{lemairePages17}), that adds weights to the traditional MLMC method in order to improve the efficiency, especially for those cases where $f$ is irregular. Our contributions are as follows: first, we provide a proof that antithetic sampling of MLMC levels is always efficient in the context of nested Monte Carlo, regardless of the regularity of $f$, and quantify the reduction in variance. Then we propose a numerical optimization scheme for MLMC parameters that fully accounts for $\tau$. This approach yields significant improvements in non-asymptotic performance while maintaining asymptotic theoretical properties. Computational experiments demonstrate that this numerical optimization is crucial in a setting where $f$ is an indicator function or when $\tau$ is large. Finally, we show numerically how the ML2R parametrized with our methodology provide an efficient way to estimate (\ref{eq:proba_large_loss}) and (\ref{eq:quantile_definition}) when high precisions are required.

\section{Multi-level Monte Carlo and antithetic variance reduction}

In this section, we start by introducing MLMC estimators for computing expectations of the form $I = \mathbb{E}[f(L)]$ where $f : \mathbb{R} \longrightarrow \mathbb{R}$ is any function satisfying $\mathbb{E}[|f(L)|^2] < +\infty$, presenting both the standard MLMC estimator and the ML2R variant. Then, we show that one can always construct an antithetic sampling of the MLMC levels regardless of the regularity of $f$ and quantify the reduction in variance. 

Finally, to facilitate the introduction of our extensions in the subsequent sections, we recall state of the art theoretical properties of these estimators with particular emphasis on their asymptotic performances.

\subsection{Introduction to Multi-level Monte Carlo estimators}
The standard approach for estimating $I$ in the nested MC framework is the nested MC estimator. This method involves first fixing the number of inner samples $K$, then computing an empirical average of $Y_K$ based on $J$ i.i.d samples. As mentioned previously, this estimator is sample-inefficient and incurs prohibitive computational costs when high precision is required. To obtain more efficient estimators, we can employ a MLMC method ; it consist of choosing an increasing sequence of inner sample sizes, $(K_r)_{r \in \{1, \dots, R\}}$, for some number of levels $R \in \mathbb{N}$, rather than using a single inner sample size $K$ as in the nested MC estimator. Observing that
\begin{equation*}
    \mathbb{E}[Y_{K_R}] = \mathbb{E}[Y_{K_1}] + \sum_{r = 2}^R \mathbb{E}[Y_{K_r} - Y_{K_{r-1}}]\,,
\end{equation*}
instead of sampling $Y_{K_R}$, we can sample $Y_{K_1}$ and the differences $Y_{K_r} - Y_{K_{r-1}}$ for each level $r \in \{1, \dots, R \}$. This approach allows to allocate more samples to the lower (computationally inexpensive) levels. Indeed, since the upper-level terms involve differences of potentially correlated random variables, they tend to have lower variance. Lemaire-Pagès \cite{lemairePages17} (Definition 3.2) introduce a general MLMC estimator that encompasses both the standard MLMC and ML2R approaches. Below, we recall the definition adapted to our setting.

\begin{definition}
    Let $\Theta$ denote the parameter space of the Multi-level Monte Carlo estimator, defined as the set
    \begin{equation*}
        \Theta = \left \{(J, q, K, R) \in (0, +\infty) \times [0,1]^{\mathbb{N}} \times (0, +\infty) \times \mathbb{N} : \forall r \in \{1, \dots, R\} \; q_r > 0, \sum_{r = 1}^R q_r = 1 \right \}
    \end{equation*}
    Here,
    \begin{enumerate}
        \item $J$ is a parameter that serves as a proxy for the total number of outer samples. For each level $r$, we define $J_r = \lceil J q_r \rceil$ the number of outer samples allocated to level $r$. The actual total number of outer samples is given by $\sum_{r = 1}^R J_r$ which might exceed $J$, but the difference is typically negligible in practice,
        \item $q$ is a sequence specifying the allocation of outer samples to each of the $R$ levels (note that $q_r = 0$ when $r > R$),
        \item $K$ is a parameter that serves as a proxy for the number of inner samples used in the first level. The actual number of inner samples used is $K_1 = \lceil K \rceil$, the difference is typically negligible in practice. More generally for each level $r$, we define $K_r = \lceil K \rceil 2^{r-1}$ the number of inner samples used at level $r$.
    \end{enumerate}
    
\end{definition}

\begin{remark}
    In the above definition we choose a geometric progression for the sequence $(K_r)_{r \in \{1, \dots,  R\}}$ with ratio 2. This is a standard choice for MLMC methods within the nested Monte Carlo framework (see, e.g., \cite{giles15}). Lemaire-Pagès \cite{lemairePages17} (Theorem 3.12) show that a ratio of 2 is optimal for the ML2R when $f$ is an irregular function, otherwise other choices could theoretically be more efficient. For simplicity, we adopt the ratio 2 in all cases throughout this paper.
\end{remark}

\begin{definition}
\label{def:general_mlmc_def}
We consider a family of square-integrable random variable $(\Delta Y_{2N})_{N \in \mathbb{N}}$ such that for all $N \in \mathbb{N}$, $\mathbb{E}[\Delta Y_{2N}] = \mathbb{E}[Y_{2N} - Y_{N}]$. We also assume that $\Delta Y_{2N}$ has a sampling cost $\gamma_{\tau}(2N)$. 

Let $\theta = (J, q, K, R) \in \Theta$, for each $r \in \{1, \dots, R\}$, let $(\Delta Y_{K_r}^j)_{j \in \{ 1, \dots, J_r\}}$ be i.i.d copies of $\Delta Y_{K_r}$. Letting $(A_{r}^R)_{r \in \{1, \dots, R\}}$ be a sequence of real numbers, which may depend on $R$, we define the general MLMC estimator of $I$ as 
\begin{equation}
\label{eq:general_mlmc_def}
	\hat{I}_{\theta} = \frac{1}{J_1} \sum_{j = 1}^{J_1} Y_{K_1}^{j} + \sum_{r = 2}^{R} \frac{A^R_{r}}{J_r} \sum_{j = 1}^{J_r} \Delta Y_{K_r}^{j} \,.
\end{equation}
\end{definition}
Different choices for the sequence of weights $(A^R_r)_{r \in \{2, \dots, R\}}$ in this definition lead to different variants of the MLMC estimator. Namely the standard MLMC as introduced by Giles~\cite{gilesMLMC2008} is obtained when $A^R_{r} = 1, \; r = 1,\dots, R$. The ML2R as introduced by Lemaire-Pagès \cite{lemairePages17}, is obtained when $A^R_{r} = W_r^R, \; r = 1, \dots, R$, for some specific sequence of weights $(W_r^R)_{r \in \{1, \dots, R\}}$ (see Appendix \ref{apx:details_weights} for details about their construction). It is worth noting that when $R = 1$, we recover the standard nested estimator (\ref{eq:std_nested_mc}). In this sense, MLMC generalizes nested MC. In this paper, if we need to discuss specifically one version of the estimator, the standard MLMC will be denoted by $\hat{I}^{\textrm{MLMC}}$, while the ML2R variant will be denoted by $\hat{I}^{\textrm{ML2R}}$. We now introduce the standard choice for the family $(\Delta Y_{2N})_{N \in \mathbb{N}}$.
\begin{definition}
    For any $N \in \mathbb{N}$, we let $\Delta Y_{2N}^S$ be the standard construction of an MLMC level with $K = 2N$ inner scenarios. It is defined by $\Delta Y^{S}_{2N} := Y^f_{2N} - Y_{2N}^c$ where
    \begin{equation*}
        Y^f_{2N} := f(\hat{E}^f_{2N}(X)), \quad Y_{2N}^c := f(\hat{E}^c_{2N}(X)) \,,        
    \end{equation*}
    with $\hat{E}^f_{2N}(X)$ and $\hat{E}^c_{2N}(X))$ defined as
    \begin{equation*}
        \hat{E}^f_{2N}(X) := \frac{1}{2N} \sum_{n = 1}^{2N} F(X, U_{n}), \quad         \hat{E}^c_{2N}(X) := \frac{1}{N} \sum_{n = 1}^{N} F(X, U_{n}) \,,
    \end{equation*}
    where $(U_{n})_{n \in \{1, \dots, N\}}$ is an i.i.d samples generated from $U$.
\end{definition}
Notice that in the above definition, the first $N$ samples of $U$ from the fine level estimator $\hat{E}^f_{2N}(X)$ are re-used in constructing the coarse level estimator $\hat{E}^c_{2N}(X)$. This approach, as opposed to using independent samples for each level, provides two key advantages. First, it reduces the variance of $\Delta Y^S_{2K}$ by increasing the correlation between the term of the difference. Second, it lowers the sampling cost from $\gamma_{\tau}(2N) + \gamma_{\tau}(N)$ to $\gamma_{\tau}(2N)$ as only $2N$ inner samples are required. It is straightforward to check that $\Delta Y^{S}_{2N}$ is square integrable and that $\mathbb{E}[\Delta Y^{S}_{2N}] = \mathbb{E}[Y_{2N} - Y_{N}]$, therefore fulfilling the conditions of Definition \ref{def:general_mlmc_def}. \\

The computational cost of an MLMC estimators with $R \in \mathbb{N}$ levels is the sum of the computational cost of each level $r \in \{1, \dots, R\}$, where we assigned an outer simulation budget of $J_r = \lceil J q_r \rceil$ with an inner simulation budget of $K_r = \lceil K 2^{r-1}\rceil$.

\begin{definition}
    Let $\theta = (J, q, K, R) \in \Theta$, we define the computational cost of $\hat{I}_{\theta}$ as $\mathcal{C}_{\tau}(\theta) := \sum_{r = 1}^R \lceil J  q_r \rceil \gamma_{\tau}(K_r)$.
\end{definition}

The cost function defined above is not particularly tractable. However, in practice, $\lceil J q_r \rceil$ is typically very close to $J q_r$. Therefore, it is reasonable to use the following approximation:
\begin{equation*}
    \mathcal{C}_{\tau}(\theta) \approx J  \times \sum_{r = 1}^R q_r \gamma_{\tau}(K_r)
\end{equation*}
In this expression, $\sum_{r = 1}^R q_r \gamma_{\tau}(K_r)$ can be interpreted as the computational cost per outer sample.

\begin{definition}
    Let $\theta = (J, q, K, R) \in \Theta$,  define the computational cost per outer-samples of $\hat{I}_{\theta}$ by
    \begin{equation}
        \label{eq:kappa_tau}\mathcal{\kappa}_{\tau}(q, K, R) := \sum_{r = 1}^R q_r \gamma_{\tau}(K_r) \,.
    \end{equation}
    and define the approximate computational cost of $\hat{I}_{\theta}$ by
    \begin{equation}
    \label{eq:tilde_c_tau}
        \tilde{\mathcal{C}}_{\tau}(\theta) := J   \mathcal{\kappa}_{\tau}(q, K, R) \,.
    \end{equation}
\end{definition}

As is standard for MLMC estimators, we will measure its precision with the Root Mean Squared Error (RMSE), i.e the $L^2$-norm between the target and the MLMC estimation.
\begin{definition}
\label{def:mse_theta}
     Let $\theta \in \Theta$, the Mean Squared Error (MSE) of $\hat{I}_{\theta}$ defined by $\mathcal{M}(\theta) = \mathbb{E}[(I - \hat{I}_{\theta})^2] := \mu(\theta)^2 + \mathcal{V}(\theta)$ where $\mu(\theta):= \mathbb{E}[I - \hat{I}_{\theta}]$ is the bias of the estimator and $\mathcal{V}(\theta) := \mathrm{Var}[\hat{I}_{\theta}] $ is the variance.
\end{definition}

For convenience we introduce the following notations.
\begin{definition}
    Let $\theta = (J, q, K, R) \in \Theta$, for each $r \in \{1, \dots, R\}$, we let $\sigma^2(r, K)$ be the structural variance of the level $r$ defined by
    \begin{equation*}
    \sigma^2(r, K) :=
        \begin{cases}
            \mathrm{Var}[Y_{K_1}] & r = 1\\
            (A^R_{r})^2\mathrm{Var}[\Delta Y_{K_r}] & r \in \{ 2, \dots, R \}
        \end{cases}
    \end{equation*}
\end{definition}
From the independence of the levels in (\ref{eq:general_mlmc_def}), the variance of an MLMC estimator is controlled by the sum of the variance of its levels. While the variance of a level $r \in \{1, \dots, R\}$ is controlled by the ratio between $\sigma^2(r, K)$ its structural variance and $J_r$ its number of allocated outer samples.
\begin{proposition}
\label{prop:variance_mlmc_formula}
    Let $\theta = (J, q, K, R) \in \Theta$, then
    \begin{equation}
    \label{eq:variance_mlmc}
        \mathcal{V}(\theta) = \sum_{r = 1}^{R} \frac{\sigma^2(r, K)}{J_r}
    \end{equation}
\end{proposition}
\begin{proof}
    It follows immediately from the independence of the levels in the definition of an MLMC estimator.
\end{proof}

\subsection{Antithetic variance reduction}
\label{sec:antithetic_variance_reduction}

A quick inspection of (\ref{eq:variance_mlmc}) reveals that, for a fixed computational budget, introducing additional upper levels to the traditional nested Monte Carlo estimator leads to an increase in variance. For the inclusion of these levels to be beneficial, this variance inflation must be sufficiently small so that, when combined with the reduction in bias, the overall MSE decreases. Each upper level $r \in \{2, \dots, R\}$ contributes a variance term proportional to $\mathrm{Var}[\Delta Y_{K_r}]$, therefore, ensuring that $\mathrm{Var}[\Delta Y_{K_r}]$ remains as small as possible is essential for the effectiveness of MLMC estimators.

\begin{remark}
\label{rem:mlmc_variance_indicator}
    In contexts where $f$ in (\ref{eq:I_def}) is smooth, for $r \geq 2$, $\mathrm{Var}[\Delta Y_{K_r}]$ is often an order of magnitude smaller than $\mathrm{Var}[Y_{K_1}]$, so the contribution of the upper-levels to the overall variance is already quite limited. However, when $f$ is an indicator function, it is common for $\mathrm{Var}[\Delta Y_{K_r}]$, to be the same order of magnitude as $\mathrm{Var}[Y_{K_1}]$. Thus, reducing $\mathrm{Var}[\Delta Y_{K_r}]$ becomes even more critical in cases where $f$ is an indicator function.
\end{remark}

A common technique found in the MLMC literature for reducing $\mathrm{Var}[\Delta Y_{K_r}]$ when $f$ is smooth is to construct $\Delta Y_{K_r}$ as an antithetic variable of $Y_{K_{r}} - Y_{K_{r-1}}$ by reusing all the samples of the fine level $Y_{K_r}$ in two conditionally independent copies of the coarse level $Y_{K_{r-1}}$ ; this is possible since $K_r = 2 K_{r-1}$ (see for example \cite{Bujok2013} or \cite{haji-ali12}).

\begin{definition}
    For any $N \in \mathbb{N}$, let
    \begin{equation*}
        \hat{E}^{f}_{2N}(X) := \frac{1}{2N} \sum_{n = 1}^{2N} F(X, U_n), \quad \hat{E}^{c}_{2N}(X) := \frac{1}{N} \sum_{n = 1}^{N} F(X, U_n) , \quad         \hat{E}^{c'}_{2N}(X) := \frac{1}{N} \sum_{n = N+1}^{2N} F(X, U_n) \,,
    \end{equation*}
    with $(U_{n})_{n \in \{1, \dots, N\}}$ an i.i.d samples generated from $U$. We then define the fine level $Y^{f}_{2N} := f(\hat{E}^{f}_{2N}(X))$, and the two conditionally coarse levels $Y^{c}_{2N} := f(\hat{E}^{c}_{2N}(X))$ and $Y^{c'}_{2N} := f(\hat{E}^{c'}_{2N}(X))$.
\end{definition}

It is well known (see, e.g., Pagès \cite{numProbaPages}, Proposition 9.3) that if $f$ is smooth enough and if the variance of $Y_{2K} - Y_{K}$ converges to $0$ when $K \rightarrow +\infty$, then the antithetic construction improves the rate of variance convergence. More precisely, briefly anticipating the assumption $(\mathrm{Var}_{\beta})$ (formally introduced in Section \ref{sec:mlmc_theo_guarantees}) : in the smooth case, this assumption typically holds with a rate of variance decay $\beta = 1$, while with antithetic sampling it is satisfied for $\beta > 1$, ensuring a faster convergence of the variance to 0 when $K$ grows. However when $f$ is not smooth (such as in the case of an indicator function), this theoretical guarantee generally no longer apply. Nevertheless, some authors still apply antithetic sampling in non-smooth cases (see, \cite{gilesAli19}, \cite{HajiAli2023NestedMM}).

\begin{definition}
    For any $N \in \mathbb{N}$, we let $\Delta Y^A_{2N}$ be the antithetic construction of an MLMC level with $2N$ inner samples defined by,
    \begin{equation}
    \label{eq:antithetic_sampling}
        \Delta Y^{A}_{2N} := Y^{f}_{2N} - \frac{1}{2}(Y^{c}_{2N} + Y^{c'}_{2N})
    \end{equation}
\end{definition}
Clearly, as required by Definition \ref{def:general_mlmc_def}, the sampling cost for $\Delta Y^{A}_{2N}$ is $\gamma_{\tau}(2N)$ since the total number of samples used is still $2N$. The other requirements, namely the square-integrability and expectation property are straightforward to check.

In Theorem \ref{thm:systematic_antithetic}, we provide a general theoretical justification showing that antithetic sampling yields variance reduction regardless of $f$ and its smoothness, and quantify this reduction. As far as we are aware, this result is new.
\begin{theorem}
\label{thm:systematic_antithetic}
    For any $N \in \mathbb{N}$,
    \begin{equation*}
        \mathrm{Var}[\Delta Y_{2N}^A] = \mathrm{Var}[\Delta Y_{2N}^S] - \frac{1}{2} \mathbb{E}[\mathrm{Var}[Y_{N} | X]]
    \end{equation*}
    or equivalently,
    \begin{equation*}
        \mathrm{Var}[\Delta Y_{2N}^S] - \mathrm{Var}[\Delta Y_{2N}^A] = \frac{1}{2} ||  Y_{N} - \mathbb{E}[Y_{N} | X] ||^2_{L_2} = \frac{1}{2} \mathbb{E}[\mathrm{Var}[Y_{N} | X]]
    \end{equation*}
    This difference is always non-negative and it is positive whenever $F(X, U)$ is not independent from $X$.
\end{theorem}
\begin{proof}
    Let $N \in \mathbb{N}$, then
    \begin{equation*}
        \mathrm{Var}[Y_{2N}^S] = \mathrm{Var}[Y_{2N}] + \mathrm{Var}[Y_{N}] - 2 \text{Cov}(Y^{f}_{2N}, Y^{c}_{N})
    \end{equation*}
    and $\mathrm{Var}[\Delta Y_{2N}^A] = \mathrm{Var}[Y_{2N}] + \frac{1}{4} \mathrm{Var}[Y^{c}_{N} + Y^{c'}_{N}] - \text{Cov}(Y_{2N}, Y_{N}^{c}) -  \text{Cov}(Y_{2N}, Y^{c'}_{N})$.
Since $\hat{E}_N^{c}$ and $\hat{E}_N^{c'}$ have the same law, are independent conditionally to $X$ and play a symmetric role in  $\text{Cov}(Y_{2N}, Y^{c}_{N})$ we have the identity $\text{Cov}(Y_{2N}, Y^{c}_{N}) = \text{Cov}(Y_{2N}, Y^{c'}_{N})$, which gives $\mathrm{Var}[\Delta Y_{2N}^A] = \mathrm{Var}[Y_{2N}] + \frac{1}{4} \mathrm{Var}[Y^{c}_{N} + Y^{c'}_{N}] - 2\text{Cov}(Y_{2N}, Y_{N}^{c})$
    and therefore $\mathrm{Var}[\Delta Y_{2N}^{S}] - \mathrm{Var}[\Delta Y_{2N}^{A}] = \mathrm{Var}[Y_{N}] - \frac{1}{4} \mathrm{Var}[Y_{N}^{c} + Y^{c'}_{N}]$. Now expanding $\mathrm{Var}[Y_{N}^{c} + Y^{c'}_{N}]$ in the above equation gives,
    \begin{equation*}
         \mathrm{Var}[\Delta Y_{2N}^{S}] - \mathrm{Var}[\Delta Y_{2N}^{A}] = \frac{1}{2} \left( \mathrm{Var}[Y_{N}] - \text{Cov}(Y^{c}_{N}, Y^{c'}_{N}) \right) = \frac{1}{2} \left( \mathbb{E}[Y_{N}^2] - \mathbb{E}[Y^{c}_{N} Y^{c'}_{N}] \right) \,.
    \end{equation*}

    From tower property of conditional expectations together with the conditional independence of $Y^{c}_{N}$ and $Y^{c'}_{N}$ we get $\mathbb{E}[Y^{c}_{N} Y^{c'}_{N}] = \mathbb{E}[(\mathbb{E}[Y_{N} | X])^2]$ which finally gives,
    \begin{equation*}
        \mathrm{Var}[\Delta Y_{2N}^{S}] - \mathrm{Var}[\Delta Y_{2N}^{A}] = \frac{1}{2} \left( \mathbb{E}[Y_{N}^2] - \mathbb{E}[(\mathbb{E}[Y_{N} | X])^2]\right) = \frac{1}{2}  || Y_{N} - \mathbb{E}[Y_{N} | X] ||^2_{L_2} \,.
    \end{equation*}
    This completes the proof.
\end{proof}
Since the case where $f$ is an indicator function will be of special interest later in this paper, it is worth applying the above proposition to this special case :
\begin{corollary}
    Let $u \in \mathbb{R}$, if for all $x \in \mathbb{R}$, $f(x) = \mathbbm{1}_{x \leq u}$ then,
    \begin{equation*}
        \mathrm{Var}[\Delta Y_{2N}^{S}] - \mathrm{Var}[\Delta Y_{2N}^{A}] = \frac{1}{2} \mathbb{E}[\mathbb{P}(\hat{E}_N \leq u | X)(1 - \mathbb{P}(\hat{E}_N \leq u | X))]
    \end{equation*}
\end{corollary}
\begin{proof}
    It suffices to remark that when $f$ is an indicator function $\mathrm{Var}[Y_N | X] = \mathbb{P}(\hat{E}_N \leq \eta | X)(1 - \mathbb{P}(\hat{E}_N \leq \eta | X))$
\end{proof}
Note that since $\Delta Y^A_N$ and $\Delta Y^S_N$ have equal sampling costs, antithetic sampling is always beneficial in a nested MC setting, regardless of the regularity of $f$. We will therefore use the antithetic version for upper levels and omit the superscript $A$ in the notation. 

\subsection{Theoretical guarantees on asymptotic complexity of MLMC estimators}
\label{sec:mlmc_theo_guarantees}
In this section, to establish the context and notation for our subsequent theoretical developments and improvements, we review standard results from Lemaire-Pagès \cite{lemairePages17} on the asymptotic complexity of both traditional and weighted MLMC estimators. Specifically, under assumptions controlling the bias and variance of the estimators, they quantify the asymptotic computational cost required to obtain an estimate of $I$ with an RMSE $\varepsilon$ as $\varepsilon \to 0$. These results demonstrate how MLMC methods can achieve significant computational savings compared to the traditional nested simulation approach, especially as higher accuracy is desired. They also highlight how the weighted MLMC estimator outperform the standard MLMC estimator in scenarios where $f$ is an irregular function. \\

Below, we recall key assumptions required for the complexity theorems.
\begin{assumption}[Weak Error Expansion]
\label{ass:weak-error}
We say that the weak error expansion \((\mathrm{WE}_{\alpha, R})\) holds for some integer \(R \geq 1\) and coefficient \(\alpha > 0\) if there exist real constants \(c_1, c_2, \dots, c_R\) and a function \(\eta_R : \mathbb{N} \to \mathbb{R}\) such that \(\eta_R(N) \to 0\) as \(N \to \infty\), and
\begin{equation*}
    \mathbb{E}[Y_{N}] = \mathbb{E}[Y] + \sum_{r = 1}^R \frac{c_r}{N^{\alpha r}} + \frac{1}{N^{\alpha R}} \cdot \eta_R(N)  \,.
\end{equation*}

We write \((\mathrm{WE}_{\alpha})\) when the above expansion holds for all \(R \geq 1\) and a common sequence of constants $(c_r)_{r \in \mathbb{N}}$. In that case, we define $\tilde{c}_{\infty} := \lim_{R \to \infty} c_R^{1/R}$ assuming that the limit exists.
\end{assumption}

\begin{assumption}
\label{ass:var_decay}
    We say that the variance decay assumption $(\mathrm{Var}_{\beta})$ holds with coefficient $\beta > 0$ if there exist a constant $V_1 > 0$ such that
    \begin{equation*}
        \forall N \in \mathbb{N}, \quad \mathrm{Var}[\Delta Y_{2N}] \leq \frac{V_1}{(2N)^{\beta}} \,.
    \end{equation*}
    and a constant $\bar{\sigma} > 0$ such that
    \begin{equation*}
        \forall N \in \mathbb{N}, \quad \mathrm{Var}[Y_{N}] \leq \bar{\sigma}^2
    \end{equation*}
\end{assumption}

The two theorems below are adapted from Lemaire-Pagès \cite{lemairePages17} (Theorem 3.12), with changes in notation for consistency with our framework.
\begin{theorem}[\cite{lemairePages17}, Theorem 3.12, a)]
\label{thm:complexity_ML2R}
Assume $\tau = 0$. Consider the case of of the weighted Multi-level estimator $\hat{I}^{\mathrm{ML2R}}$. Assume that $(\mathrm{WE}_{\alpha})$ holds for some $\alpha > 0$ and $(\mathrm{Var}_{\beta})$ for some $\beta > 0$. Assume furthermore that
\begin{equation*}
    \underset{R \in \mathbb{N}}{\sup} \;\underset{k \in \mathbb{N}}{\sup} \; |\eta_R(k)| <+\infty    
\end{equation*}
and $\tilde{c}_{\infty} > 0$. Then, there exists a $\varepsilon_0 > 0$, a collection $\{ \theta_{\varepsilon} \in \Theta : \varepsilon \in (0, \varepsilon_0) \}$ and a $C_{\mathrm{ML2R}} \geq 0$ verifying
\begin{equation*}
    \underset{\varepsilon \to 0}{\lim \sup} \; v(\varepsilon, \beta) \cdot \tilde{\mathcal{C}}_{0}(\theta_{\varepsilon}) \leq C_{\mathrm{ML2R}}
\end{equation*}
where,
\begin{equation}
\label{eq:v_complexity_ml2r}
    v(\varepsilon, \beta) = \begin{cases}
        \varepsilon^2, & \text{if} \quad \beta > 1, \\
        \varepsilon^2(\log(1/\varepsilon))^{-1}, & \text{if} \quad \beta = 1, \\
        \varepsilon^2e^{-\frac{1 - \beta}{\sqrt{\alpha}}\sqrt{2\ln(1/\varepsilon)\ln(2)}}, & \text{if} \quad \beta < 1.
    \end{cases}
\end{equation}
and 
\begin{equation*}
    \underset{\varepsilon \to 0}{\lim \sup} \; \varepsilon^{-2} \mathcal{M}(\theta_{\varepsilon}) \leq 1
\end{equation*}
\end{theorem}

\begin{theorem}[\cite{lemairePages17}, Theorem 3.12 b]
\label{thm:complexity_mlmc}
Assume $\tau = 0$. Consider the case of of the standard Multi-level estimator $\hat{I}^{\mathrm{MLMC}}$. Assume that $(\mathrm{WE}_{\alpha, 1})$ hold for some $\alpha > 0$ and $(\mathrm{Var}_{\beta})$ for some $\beta > 0$. Assume furthermore that $c_1 \neq 0$. Then, there exists a $\varepsilon_0 > 0$, a collection $\{ \theta_{\varepsilon} \in \Theta : \varepsilon \in (0, \varepsilon_0) \}$ and a $C_{\mathrm{MLMC}} \geq 0$ verifying
\begin{equation*}
    \underset{\varepsilon \to 0}{\lim \sup} \; v(\varepsilon, \beta) \cdot \tilde{\mathcal{C}}_0(\theta_{\varepsilon}) \leq C_{\mathrm{MLMC}}
\end{equation*}
where,
\begin{equation}
\label{eq:v_complexity_mlmc}
    v(\varepsilon, \beta) = \begin{cases}
        \varepsilon^2, & \text{if} \quad \beta > 1, \\
        \varepsilon^2(\log(1/\varepsilon))^{-2}, & \text{if} \quad \beta = 1, \\
        \varepsilon^{2 + \frac{1 - \beta}{\alpha}}, & \text{if} \quad \beta < 1.
    \end{cases}
\end{equation}
and
\begin{equation*}
    \underset{\varepsilon \to 0}{\lim \sup} \; \varepsilon^{-2} \mathcal{M}(\theta_{\varepsilon}) \leq 1
\end{equation*}
\end{theorem}

These theorems reveal three regimes of asymptotic complexity, determined by the variance decay rate~$\beta$. When $\beta > 1$, both estimators achieve the target precision~$\varepsilon$ with a computational complexity of $O(\varepsilon^{-2})$, which matches the optimal complexity attainable if unbiased samples of $L$ were available. In the intermediate regime, where $\beta = 1$, both estimators require a computational cost of order $O(\varepsilon^{-2} \log(\varepsilon)^c)$ for some constant $c>0$; that is, the complexity is essentially $O(\varepsilon^{-2})$ up to a logarithmic factor. In this case, the ML2R estimator already outperforms the standard version. When $\beta < 1$, the complexity departs further from $O(\varepsilon^{-2})$. Specifically, the ML2R estimator achieves a complexity of $O(\varepsilon^{-2-\nu})$ for any $\nu > 0$, while the standard MLMC estimator has complexity $O\left(\varepsilon^{-2 - \frac{1-\beta}{\alpha}}\right)$. In all cases, MLMC methods provide a significant asymptotic improvement over the traditional nested simulation approach.

\begin{remark}
    The parameter $\beta$ is strongly related to the smoothness of the underlying function $f$ where more smoothness yields higher $\beta$. Typically, when $f$ is an indicator function $\beta = \frac{1}{2}$ (see \cite{gilesAli19} or \cite{Giorgi_2020})). The regime where $\beta < 1$ is precisely the case in which the ML2R estimator is significantly more efficient than the standard MLMC, making it the preferred estimator for estimating (\ref{eq:proba_large_loss}). Numerical evidence in Section \ref{sec:num_appli} will back-up this claim.
\end{remark}

In Theorem \ref{thm:complexity_any_tau} we will extend these results to arbitrary $\tau \geq 0$, replacing $\tilde{\mathcal{C}}_{0}$ by $\tilde{\mathcal{C}}_{\tau}$. The asymptotic rates are still achieve by the standard parametrization $(\theta_{\varepsilon})_{\varepsilon \in (0, \varepsilon_0)}$. However, for non-asymptotic regimes, the parametrization is sub-optimal. To address this, Section \ref{sec:optimizing_parameters} introduces an improved parametrization that enhances performance in non-asymptotic settings.

\section{Optimizing parameters selection for MLMC in non-asymptotic regimes}
\label{sec:optimizing_parameters}

The objective of this section is to present a new parametrization of MLMC estimators that explicitly incorporates the parameter $\tau$ and improve complexity in non-asymptotic regimes. As demonstrated by our numerical experiments in Section \ref{sec:num_appli}, this novel parametrization improves performance of MLMC when applied to a nested Monte Carlo framework with indicator function payoffs.

We first review the standard parametrization achieving the computational complexities of Theorem \ref{thm:complexity_ML2R} and \ref{thm:complexity_mlmc} ; then prove it generalizes to context with generic $\tau \geq 0$. Building on this foundation, we introduce a novel parametrization involving a numerical optimization for parameters $K$ and $R$ fully accounting for generic $\tau$ and improving the computational cost for non-asymptotic regimes while preserving theoretical guarantees of Theorem \ref{thm:complexity_ML2R} and \ref{thm:complexity_mlmc}. Finally, we discuss asymptotics of the optimized parameters when $\tau \to +\infty$. 

\subsection{State-of-the-art MLMC parameterization}

\subsubsection{Plug-and-play parameters}
Traditionally, as developed by Giles (see \cite{giles15}, Algorithm 1), the determination of the optimized parameters $\theta_{\varepsilon}$ for a prescribed precision $\varepsilon > 0$ in the standard MLMC setting is based on an adaptive algorithm. This algorithm begins with an initial choice for the number of levels and the number of outer samples per level, then iteratively refine these parameters based on ongoing estimates of the bias and variance. In contrast, the method proposed by Lemaire-Pagès is more "plug-and-play" (see \cite{lemairePages17} Practitioner’s corner 5.1): after a preliminary phase in which certain structural constants of the problem are estimated, optimal parameters can be computed in closed form. Here we will discuss exclusively the approach of Lemaire-Pagès. \\

The objective is to solve, at least approximately, the following minimization problem
\begin{equation}
\label{eq:optim_param_base}
    \underset{\substack{\theta \in \Theta \\ \mathcal{M}(\theta) \leq \varepsilon^2}} {\min}\tilde{\mathcal{C}}_{\tau}(\theta)\,.
\end{equation}
Lemaire-Pagès \cite{lemairePages17} deal with the asymptotic regime as $\varepsilon \to 0$, with $\tau = 0$. They construct a collection of closed-form optimized parameters $\{\theta_{\varepsilon} : \varepsilon \in (0, \varepsilon_0]\}$ for some $\varepsilon_0 > 0$ such that, under the hypothesis of Theorem \ref{thm:complexity_ML2R} (resp. Theorem \ref{thm:complexity_mlmc}) for the ML2R case (resp. standard MLMC case),
\begin{equation*}
    \underset{\varepsilon \to 0}{\lim \sup} \; \varepsilon^{-2} \mathcal{M}(\theta_{\varepsilon}) \leq 1 \,,
\end{equation*}
and
\begin{equation*}
    \lim \sup v(\varepsilon, \beta) \cdot\tilde{\mathcal{C}}_0(\theta_{\varepsilon}) \leq C_{\alpha, \beta} \,,
\end{equation*}
where $C_{\alpha, \beta}$ is a finite constant depending on $\alpha, \beta$ and the type of estimator (standard MLMC or ML2R) and $v(\varepsilon, \beta)$ is defined in (\ref{eq:v_complexity_ml2r}) and (\ref{eq:v_complexity_mlmc}). We report in Table \ref{tab:closed_mlmc_parameters_v2} the explicit values of these parameters. We refer to \cite{lemairePages17} (see in particular 5.1 Practitionner's corner) for their construction in our nested MC context with a strong-error assumption of type $(\mathrm{Var}_{\beta})$.
\begin{table}[H]
    \centering
    \renewcommand{\arraystretch}{3}
    \begin{tabular}{c | c | c}
        Parameter & Weighted (ML2R) & Classical (MLMC)\\
        \hline
        $R(\varepsilon)$ & $\left \lceil \frac{1}{2} + \ln_2(\frac{\tilde{c}^{\frac{1}{\alpha}}}{\underline{K}}) + \sqrt{(\frac{1}{2} + \ln_2(\frac{\tilde{c}^{\frac{1}{\alpha}}}{\underline{K}}))^2 + \frac{2 \ln_2(\varepsilon^{-1}\sqrt{1 + 4\alpha})}{\alpha}} \right \rceil$ & $\left \lceil 1 + \ln_2(\frac{|c_1|^{\frac{1}{\alpha}}}{\underline{K}}) + \frac{\ln_2(\sqrt{1 + 2\alpha} \varepsilon^{-1})}{\alpha} \right \rceil$ \\
        \hline
        $K^+(\varepsilon)$ & $(1 + 2\alpha R)^{\frac{1}{2 \alpha R}} \varepsilon^{-\frac{1}{\alpha R}} \tilde{c}^{\frac{1}{\alpha}} 2^{-\frac{R-1}{2}}$ & $(1 + 2\alpha)^{\frac{1}{2 \alpha}} \varepsilon^{-\frac{1}{\alpha}} |c_1|^{\frac{1}{\alpha}}2^{-(R-1)}$ \\
        \hline
        $K(\varepsilon)$ & \multicolumn{2}{|c}{$\underline{K}  \left \lceil \frac{K^+(\varepsilon)}{\underline{K}}  \right\rceil$} \\
        \hline
        $q(\varepsilon)$ & \multicolumn{2}{|c}{$q_1(\varepsilon) = \frac{\bar{\sigma}}{\mu^*_{\varepsilon}}\quad $ $\forall r = 2, \dots, R(\varepsilon) \quad q_r(\varepsilon) = \frac{\sqrt{V_1}|A_{r, R(\varepsilon)}|}{ \underline{K}^{\frac{\beta}{2}}2^{\frac{(1+\beta)(r-1)}{2}}\mu^*_{\varepsilon}}$} \\
         &  \multicolumn{2}{|c}{with $\mu_{\varepsilon}$ such that $\sum_{r = 1}^{R(\varepsilon)} q_r(\varepsilon) = 1$} \\
         \hline
         $J(\varepsilon)$ & \multicolumn{2}{|c}{$M_{\varepsilon} \frac{\frac{\bar{\sigma}^2}{q_1(\varepsilon)} + \sum_{r =2}^R \frac{(A_{r}^{R(\varepsilon)})^2 V_1}{q_r(\varepsilon) K(\varepsilon)^{\beta} 2^{(r-1)\beta}}}{\varepsilon^2} = M_{\varepsilon} \frac{\mu_{\varepsilon} \left ( \bar{\sigma} + \sqrt{V_1} K(\varepsilon)^{-\frac{\beta}{2}} \sum_{r = 2}^{R(\varepsilon)} |A_{r, R(\varepsilon)}| 2^{\frac{(r-1)(1 - \beta)}{2}} \right)}{\varepsilon^2}$}  \\
         &  \multicolumn{2}{|c}{with $M_{\varepsilon} = (1 + \frac{1}{2\alpha R(\varepsilon)})$ in the ML2R case and $M_{\varepsilon} = (1 + \frac{1}{2\alpha})$ in the MLMC case.} \\
    \end{tabular}
    \vspace{0.5cm}
    \caption{Closed-form optimized parameters for the standard and weighted Multi-level estimators. Where $\tilde{c}$ is a constant such that $\tilde{c} > \tilde{c}_{\infty}$ in $(WE_{\alpha})$, and $\underline{K}$ is a fixed integer greater or equal to 1.
    \label{tab:closed_mlmc_parameters_v2}}
\end{table}

\subsubsection{Extension of complexity theorems for arbitrary $\tau$}
\label{sec:extension_complexity_thms}
The objective of this section is to prove Theorem \ref{thm:complexity_any_tau} below. That is, we want to show that the asymptotics from Theorem \ref{thm:complexity_ML2R} and Theorem \ref{thm:complexity_mlmc} on the complexity of parameters in Table \ref{tab:closed_mlmc_parameters_v2} are preserved for arbitrary $\tau \geq 0$.
\begin{theorem}
\label{thm:complexity_any_tau}
    For any real-valued function $v$ defined on the right neighbourhood of $0$ such that $\underset{\varepsilon \to 0}{\lim \sup} \; v(\varepsilon) \tilde{\mathcal{C}}_{0}(\theta_{\varepsilon}) < +\infty$
    then,
    \begin{equation*}
        \underset{\varepsilon \to 0}{\lim \sup} \; v(\varepsilon) \tilde{\mathcal{C}}_{\tau}(\theta_{\varepsilon}) \leq \mathcal{K}_{\beta} \; \underset{\varepsilon \to 0}{\lim \sup} \; v(\varepsilon) \tilde{\mathcal{C}}_{0}(\theta_{\varepsilon})
    \end{equation*}
    where,
    \begin{equation*}
    \mathcal{K}_{\beta} :=
        \begin{cases}
             1 & \beta \leq 1 \\
            \frac{\tau}{d_{\infty}} + 1 & \beta > 1
        \end{cases}
    \end{equation*}
    Therefore the complexities of Theorem \ref{thm:complexity_ML2R} and Theorem \ref{thm:complexity_mlmc} still hold for generic $\tau \geq 0$ (when replacing $\tilde{\mathcal{C}}_{0}$ with $\tilde{\mathcal{C}}_{\tau}$), and are still achieved with parameters in Table \ref{tab:closed_mlmc_parameters_v2}.
\end{theorem}
To prove the theorem, we begin with a proposition relating $\tilde{\mathcal{C}}_{\tau}$ to $\tilde{\mathcal{C}}_0$ for generic $\tau \geq 0$.
\begin{proposition}
\label{prop:cond_equiv_cost_tau}
    Let $\tau \geq 0$ and $\theta \in \Theta$ be fixed. Then 
    \begin{equation*}
        \frac{\tilde{\mathcal{C}}_{\tau}(\theta)}{\tilde{\mathcal{C}}_{0}(\theta)} =  1 + \frac{\tau}{\lceil K \rceil \sum_{r=1}^{R} q_{r} 2^{r-1}} \,.
    \end{equation*}
\end{proposition}
\begin{proof}
    Let $\tau \geq 0$ and $\theta \in \Theta$ be fixed. From (\ref{eq:computational_cost_YK}), for all $r \in\{1, \dots, R\}$
    \begin{equation}
        \gamma_{\tau}(K_{r}) = \tau + \lceil K \rceil 2^{r-1} = \tau + \gamma_0(K_{r}) \,.
    \end{equation}
    Then using (\ref{eq:kappa_tau}), (\ref{eq:tilde_c_tau}) and the above equation we get,
    \begin{equation*}
        \tilde{\mathcal{C}}_{\tau}(\theta) =  J \sum_{r=1}^{R} q_{r} \gamma_{\tau}(K_{r}) = \tau J \sum_{r=1}^{R} q_{r} + \sum_{r=1}^{R} q_{r} \lceil K \rceil 2^{r-1} = \tau J + \tilde{\mathcal{C}}_0(\theta) \,.
    \end{equation*}
    where we used $\sum_{r=1}^R q_{r} = 1$ in the last equality. Therefore,
    \begin{equation*}
        \frac{\tilde{\mathcal{C}}_{\tau}(\theta)}{\tilde{\mathcal{C}}_{0}(\theta)} = \frac{\tau J}{\tilde{C}_0(\theta)} + 1 = \frac{\tau}{\sum_{r=1}^{R} q_{r} \lceil K \rceil 2^{r-1}} + 1
    \end{equation*}
    proving the result.
\end{proof}

We observe from this proposition that the relative behavior of $\tilde{\mathcal{C}}_{\tau}(\theta_{\varepsilon})$ and $\tilde{\mathcal{C}}_{0}(\theta_{\varepsilon})$ depends on the asymptotic behavior of $\lceil K(\varepsilon)\rceil \sum_{r=1}^{R(\varepsilon)} q_{r}(\varepsilon) 2^{r-1}$, which we analyze in the following proposition. The proof of this result is based on technical lemmas presented in Appendix \ref{apx:technical_lemmas_sec}.
\begin{proposition}
\label{prop:cv_cond_equiv}
    For all $\varepsilon > 0$, let $\theta_{\varepsilon}$ be the MLMC parameters of Table \ref{tab:closed_mlmc_parameters_v2}. If $\beta > 1$ then there exists $d_{\infty} \in (0, +\infty)$ such that
    \begin{equation}
    \label{eq:cv_somme_q_2_b_1}
        K({\varepsilon}) \sum_{r = 1}^{R(\varepsilon)} q_{r}(\varepsilon) 2^{r-1} \underset{\varepsilon \to 0}{\longrightarrow} d_{\infty} \,,
    \end{equation}
    while if $\beta \leq 1$,
    \begin{equation}
    \label{eq:cv_somme_q_2_b_2}
        K(\varepsilon) \sum_{r = 1}^{R(\varepsilon)} q_{r}(\varepsilon) 2^{r-1} \underset{\varepsilon \to 0}{\longrightarrow} +\infty
    \end{equation}
\end{proposition}
\begin{proof}
    From Lemma \ref{lemma:K_underline} we can consider $\varepsilon > 0$ to be sufficiently small such that $K({\varepsilon}) = \underline{K}$.
    Then from the definition of $q(\varepsilon)$,
    \begin{equation*}
    q_{r}(\varepsilon) =
        \begin{cases}
            \frac{\bar{\sigma}}{\sqrt{\underline{K}} \mu_{\varepsilon}}  & r = 1\\
             \frac{\sqrt{V_1} |A_{r}^{R(\varepsilon)}|}{ \underline{K}^{\frac{1 +\beta}{2}}2^{\frac{(1+\beta)(r-1)}{2}}\mu_{\varepsilon}} & r \in \{2, \dots, R(\varepsilon)\}
        \end{cases}
    \end{equation*}
    Therefore,
    \begin{equation*}
        K(\varepsilon) \sum_{r = 1}^{R(\varepsilon)} q_{r}(\varepsilon) 2^{r-1} = \frac{1}{\mu_{\varepsilon}} \left[ \frac{\bar{\sigma}}{\sqrt{\underline{K}}} + \frac{\sqrt{V_1}}{\underline{K}^{\frac{1 + \beta}{2}}} \sum_{r = 2}^{R(\varepsilon)} |A_r^{R(\varepsilon)}| 2^{\frac{r-1}{2}(1 - \beta)} \right] \,.
    \end{equation*}
     Since $R(\varepsilon) \to +\infty$ as $\varepsilon \to 0$ an application of Lemma \ref{lemma:mu_cv_q} yields that $\mu_{\varepsilon}$ converges to a limit $\mu_{\infty} \in (0, +\infty)$. Therefore it suffice to study the convergence of the term 
     \begin{equation*}
         \sum_{r = 2}^{R(\varepsilon)} |A_r^{R(\varepsilon)}| 2^{\frac{r-1}{2}(1 - \beta)} \,.
     \end{equation*}
         
    Assuming $\beta > 1$, then applying Lemma \ref{lemma:main_lemma_giorgi_weights} 2. with $\gamma = -\frac{1 - \beta}{2} > 0$ gives (\ref{eq:cv_somme_q_2_b_1}). Assuming $\beta \leq 1$, remark that,
    \begin{equation*}
        \sum_{r = 2}^{R(\varepsilon)} |A_r^{R(\varepsilon)}| 2^{\frac{r-1}{2}(1 - \beta)} \geq  |A_{R(\varepsilon)}^{R(\varepsilon)}| 2^{\frac{R(\varepsilon)-1}{2}(1 - \beta)} \,.
    \end{equation*}
    Since $R(\varepsilon) \to +\infty$ and by Lemma \ref{lemma:giorgi_lim_weights}, $|A_{R(\varepsilon)}^{R(\varepsilon)}| \to 1$ as $\varepsilon \to 0$, we have $|A_{R(\varepsilon)}^{R(\varepsilon)}| 2^{\frac{R(\varepsilon)-1}{2}(1 - \beta)} \underset{\varepsilon \to 0}{\longrightarrow} +\infty$ and get (\ref{eq:cv_somme_q_2_b_2}).
    
    Finally, assuming that $\beta = 1$ ; we let $\phi$ be the function from $\mathbb{N}$ to $\mathbb{N}$ such that for all $R \in \mathbb{N}$, $\phi(R) = R - \sqrt{R}$. On one hand from Lemma \ref{lemma:giorgi_lim_weights},
    \begin{equation}
    \label{eq:lim_inf_weights_mlmc}
        \underset{\varepsilon \to 0}{\lim} \underset{j \in \{1, \dots, \phi(R(\varepsilon))\}}{\inf} |A_{j}^{R(\varepsilon)}| > 0 \,.
    \end{equation}
    On the other hand,
    \begin{equation*}
        \sum_{r = 2}^{R(\varepsilon)} |A_r^{R(\varepsilon)}| 2^{\frac{r-1}{2}(1 - \frac{\beta}{2})} \geq \sum_{r = 2}^{\phi(R(\varepsilon))} |A_r^{R(\varepsilon)}| 2^{\frac{r-1}{2}(1 - \frac{\beta}{2})}
         \geq \underset{j \in \{1, \dots, \phi(R(\varepsilon))\}}{\inf} |A_{j}^{R(\varepsilon)}| (\phi(R(\varepsilon)) - 1) \,.
    \end{equation*}
    From (\ref{eq:lim_inf_weights_mlmc}) and $\phi(R_{\varepsilon}) \to +\infty$ as $\varepsilon \to 0$ we get
    \begin{equation*}
        \underset{j \in \{1, \dots, \phi(R(\varepsilon))\}}{\inf} |A_{j}^{R(\varepsilon)}| (\phi(R(\varepsilon)) - 1) \underset{\varepsilon \to 0}{\longrightarrow} +\infty \,,
    \end{equation*}
    which gives (\ref{eq:cv_somme_q_2_b_2}).
\end{proof}

We are now ready to prove Theorem \ref{thm:complexity_any_tau}.
\begin{proof}
    From Proposition \ref{prop:cond_equiv_cost_tau} and Proposition \ref{prop:cv_cond_equiv}, if $\beta > 1$,
    \begin{equation*}
        \frac{\tilde{\mathcal{C}}_{\tau}(\theta_{\varepsilon})}{\tilde{\mathcal{C}}_0(\theta_{\varepsilon})} \underset{\varepsilon \to 0}{\longrightarrow} 1 + \frac{\tau}{d_{\infty}}
    \end{equation*}
    whereas if $\beta \leq 1$,
    \begin{equation*}
        \frac{\tilde{\mathcal{C}}_{\tau}(\theta_{\varepsilon})}{\tilde{\mathcal{C}}_0(\theta_{\varepsilon})} \underset{\varepsilon \to 0}{\longrightarrow} 1 \,.
    \end{equation*}
    Therefore in both cases,
    \begin{equation*}
        \underset{\varepsilon \to 0}{\lim \sup} \; v(\varepsilon) \tilde{\mathcal{C}}_{\tau}(\theta_{\varepsilon}) =\mathcal{K}_{\beta} \; \underset{\varepsilon \to 0}{\lim \sup} \; v(\varepsilon) \tilde{\mathcal{C}}_0(\theta_{\varepsilon}) \,,
        \end{equation*}
        proving the claim.
\end{proof}
\subsection{Beyond closed form: Numerical parameter optimization for finite-Sample efficiency}
The optimized parameters presented in Table \ref{tab:closed_mlmc_parameters_v2} are derived using several asymptotic approximations. While these approximations are generally effective and yield robust parameters, they may become problematic in non-asymptotic regimes of precision $\varepsilon > 0$. This issue typically arise when $\tau$ is large or when $f$ is an indicator function. In such scenarios, $R(\varepsilon)$ may significantly overestimate the optimal number of levels, leading to suboptimal performances. For instance, for relatively large $\varepsilon$, often a standard nested MC (i.e $R = 1$) is more efficient than a MLMC with $R \geq 2$. However the parameters in Table \ref{tab:closed_mlmc_parameters_v2} always yield $R(\varepsilon) \geq 2$. 

In this section, we revisit the derivation of optimal parameters from Lemaire-Pagès \cite{lemairePages17}. Our goal is to generalize their approach to arbitrary $\tau \geq 0$ and to obtain more robust parameter choices for non-asymptotic regimes, including the option to select $R = 1$ when appropriate. The newly optimized parameters are presented in Table \ref{tab:num_optim_mlmc_parameters}. Among them $R^*(\varepsilon)$ and $K^*(\varepsilon)$ are obtained via a fast numerical optimization procedure, while $q^*(\varepsilon)$ and $J^*(\varepsilon)$ are given in closed-form.

The main result of this section is Theorem \ref{thm:new_optimization_resulting_cost}, which establishes that the new optimized parameters reduce the overall computational cost of the estimator while preserving its asymptotic MSE behavior. In particular, this ensures that the asymptotic guarantees provided by Theorem \ref{thm:complexity_ML2R} and Theorem \ref{thm:complexity_mlmc} remain valid.

In this section, we assume that we have access to accurate values for the structural constants $\alpha$, $c_1$, $\tilde{c} > \tilde{c}_{\infty}$ from assumption $(\mathrm{WE_{\alpha}})$ and to $\beta$, $V_1$, $\bar{\sigma}$ from assumption $(\mathrm{Var}_{\beta})$. 
Recall that the objective is to solve the following minimization problem
\begin{equation}
    \underset{\substack{\theta \in \Theta \\ \mathcal{M}(\theta) \leq \varepsilon^2}} {\min}\tilde{\mathcal{C}}_{\tau}(\theta)\,.
\end{equation}
While this problem cannot be solve directly, we build upon the approach of Lemaire-Pagès \cite{lemairePages17} for the derivation of an approximate solution. 

\subsubsection{Optimization of $J$}
We begin with the optimization of $J$ while other parameters are fixed.

\begin{definition}
    Let $\Pi$ be the parameter space defined as
    \begin{equation*}
        \Pi := \left \{(q, K, R) \in [0,1]^{\mathbb{N}} \times (0, +\infty) \times \mathbb{N} : \forall r \in \{1, \dots, R\}, \;q_r > 0, \; \sum_{r = 1}^R q_r = 1\right \}
    \end{equation*}
\end{definition}
We will consider that $\theta = (J, \pi) \in \Theta$ where $\pi \in \Pi$, which is a slight abuse of notation. Letting $\pi \in \Pi$ be fixed, the goal of this step is to find an approximate solution to
\begin{equation}
\label{eq:pb_min_tilde_c_for_J}
    \underset{ \substack{J\in (0, +\infty) \\ \mathcal{M}(J, \pi) \leq \varepsilon^2}}{\min} \; \tilde{\mathcal{C}}_{\tau}(J, \pi) \,.
\end{equation}
Since the function $(J, \pi) \mapsto \mathcal{M}(J, \pi)$ is not easily manageable, we begin by introducing an alternative function $\widetilde{\mathcal{M}}$ to serve as a proxy in optimizing the parameter $J$. Observe that the bias of an MLMC estimator depends solely on the parameters $K$ and $R$. Indeed from Definition \ref{def:general_mlmc_def} for all $\theta = (J, q, K, R) \in \Theta$, 
\begin{equation*}
    \mathbb{E}[\hat{I}_{\theta}] =  \mathbb{E}[Y_{K_1}] + \sum_{r = 2}^{R} A_r^{R} \mathbb{E}[\Delta Y_{K_r}]
\end{equation*}
Therefore, although $\mu$ is formally defined on the full parameter space $\Theta$, it depends only on parameters $(K, R)$. That is, there exists a function $\tilde{\mu} : (0, +\infty) \times \mathbb{N} \longrightarrow +\infty$ such that for all $\theta = (J, K, q, R) \in \Theta$, we have $\mu(\theta) = \tilde{\mu}(K, R)$. Slightly abusing the notation, we will repeatedly use $\mu(K,R)$ instead $\tilde{\mu}(K, R)$.

To get a more manageable expression of the variance $\mathcal{V}(\theta)$, we introduce the notion of unit variance of an MLMC estimator which will provide a fully tractable upper bound under $(\mathrm{Var}_{\beta})$.
\begin{definition}
    Let $v : \Pi \longrightarrow \mathbb{R}$ be the unit variance function of an MLMC estimator, defined as
    \begin{equation*}
        \forall \pi=(q, K, R) \in \Pi, \quad v(\pi) := \sum_{r=1}^{R} \frac{\sigma^2(r, K)}{q_r}
    \end{equation*}
\end{definition}

\begin{definition}
    Let $ \theta=(J, q, K, R) \in \Theta$, we define
    \begin{equation*}
    \bar{\sigma}^2(r, K) := 
    \begin{cases}
        \bar{\sigma}_{1}^2 & r = 1 \\
        \frac{(A^R_{r})^2 V_1}{K_r^{\beta}} & r \in \{2, \dots, R \}
    \end{cases}
    \end{equation*}
    Let $\bar{v} : \Pi \longrightarrow \mathbb{R}$ be defined by
    \begin{equation}
    \label{eq:def_v_bar}
        \forall \pi = (q, K, R) \in \Pi, \quad \bar{v}(\pi) := \sum_{r = 1}^R \frac{\bar{\sigma}^2(r, K)}{q_r} \,.
    \end{equation}
    Then under $(\mathrm{Var}_{\beta})$ we have $\bar{v} \geq v$.
\end{definition}

The following proposition justify that we can use $\frac{\bar{v}(\pi)}{J}$ as an appropriate proxy for the variance of an MLMC estimator.
\begin{proposition}
    For all $\theta = (J, \pi) \in \Theta$ we have \begin{equation}
        \mathcal{V}(\theta) \leq \frac{\bar{v}(\pi)}{J}
    \end{equation}
\end{proposition}
\begin{proof}
    Let $\theta = (J, \pi) \in \Theta$. Using Proposition \ref{prop:variance_mlmc_formula},
    \begin{equation*}
        \mathcal{V}(\theta) = \sum_{r = 1}^{R} \frac{\sigma^2(r, K)}{J_r} = \sum_{r = 1}^{R} \frac{\sigma^2(r, K)}{\lceil J q_r \rceil} \leq \frac{1}{J} \sum_{r = 1}^{R} \frac{\sigma^2(r, K)}{q_r} = \frac{v(\pi)}{J} \leq \frac{\bar{v}(\pi)}{J} \,,
    \end{equation*}
    which proves the claim.
\end{proof}

Below we recall with our notations, the result from Lemaire-Pagès \cite{lemairePages17} giving an asymptotic bias expansion for a general MLMC estimator, under $(\mathrm{WE}_{\alpha})$ in the ML2R case and under $(\mathrm{WE}_{\alpha, 1})$ in the standard MLMC case.

\begin{proposition}[\cite{lemairePages17}, Propositon 3.4]
\label{prop:mlmc_bias_expansion}
    Let $R \geq 1$ and $K \in (0, +\infty)$. In the ML2R case, assume $(\mathrm{WE}_{\alpha})$ holds and that $\underset{R \in \mathbb{N}}{\sup} \;\underset{k \in \mathbb{N}}{\sup} \; |\eta_R(k)| < +\infty$. Then there exists a bounded function $\eta : \mathbb{R}^2 \longrightarrow \mathbb{R}$ such that
        \begin{equation}
        \label{eq:bias_ml2r}
        \mu(K, R) = \frac{(-1)^{R-1} c_R}{\lceil K \rceil^{\alpha R} 2^{\frac{\alpha R(R-1)}{2}}} \left(1 + \eta(\lceil K \rceil^{-1}, R) \right) 
        \end{equation}
    In the standard MLMC case, assume that $(WE_{\alpha})$ holds with $c_1 \neq 0$, then
    \begin{equation}
    \label{eq:bias_mlmc}
        \mu(K, R) =  \frac{c_1}{\lceil K \rceil^{\alpha} 2^{(R-1)\alpha}} \left(1 + \eta_1(K_R^{-1}) \right) 
    \end{equation}
\end{proposition}
It is not possible to evaluate (\ref{eq:bias_ml2r}) and (\ref{eq:bias_mlmc}) exactly, as some key quantities are not directly accessible: $\eta_1$ in the MLMC case, and $\eta$ and $c_R$ in the ML2R case. To overcome this, we approximate the unknown terms by setting $\eta \equiv 0$, $ \eta_1 \equiv 0$, and replacing $c_R$ with its asymptotic upper bound $\tilde{c}^R$. These simplifications yield the following tractable formulas for the bias of an MLMC estimator.

\begin{definition}
    For all $\pi_0 = (K, R) \in \Pi_0$, let $\tilde{\mu} : \Pi_0 \longrightarrow \mathbb{R}$ be the tractable bias of a general MLMC estimator, defined as
    \begin{equation}
    \label{eq:proxy_bias_ml2r}
        \tilde{\mu}(K, R) = \frac{(-1)^{R-1} \tilde{c}^{R}}{\lceil K \rceil^{\alpha R} 2^{\frac{\alpha R(R-1)}{2}}}
    \end{equation}
    in the ML2R case and
    \begin{equation}
    \label{eq:mu_tilde_mlmc}
        \tilde{\mu}(K, R) =  \frac{c_1}{ \lceil K \rceil^{\alpha} 2^{(R-1)\alpha}}
    \end{equation}
    in the standard MLMC case.
\end{definition}

Combining the variance proxy and the bias proxy give us the proxy for the MSE of the estimator.
\begin{definition}
    We define the tractable proxy for the MSE of the estimator as the function $\widetilde{\mathcal{M}} : \Theta \to \mathbb{R}$ such that
    \begin{equation}
    \label{eq:def_tilde_Mse}
        \forall \pi = (J, q, K, R), \quad \widetilde{\mathcal{M}}(J, \pi) = \frac{\bar{v}(\pi)}{J} + \tilde{\mu}^2(K, R) \,.
    \end{equation}
\end{definition}

For fixed $\pi \in \Pi$, instead of (\ref{eq:pb_min_tilde_c_for_J}) we solve the more tractable problem,
\begin{equation}
\label{eq:approx_pb_min_tilde_c_for_J}
    \underset{\substack{J \in (0, +\infty) \\ \widetilde{\mathcal{M}}(J, \pi) \leq \varepsilon^2}}{\min} \; \tilde{\mathcal{C}}_{\tau}(\theta)
\end{equation}

\begin{proposition}
    Problem (\ref{eq:approx_pb_min_tilde_c_for_J}) admits a solution if and only if $\pi = (q, K, R) \in \Pi$ is such that $\tilde{\mu}(K,R) < \varepsilon$. In that case the solution is given by,
    \begin{equation*}
        J(\varepsilon, \pi) = \frac{\bar{v}(\pi)}{\varepsilon^2 - \tilde{\mu}^2(K, R)}
    \end{equation*}
\end{proposition}
\begin{proof}
For any $\pi \in \Pi$, from (\ref{eq:tilde_c_tau}), $J \mapsto \tilde{\mathcal{C}}_{\tau}(J, \pi)$ is increasing. While from (\ref{eq:def_tilde_Mse}), $J \mapsto \widetilde{\mathcal{M}}(J, \pi)$ is decreasing. Therefore (\ref{eq:approx_pb_min_tilde_c_for_J}) is solved for the $J$ saturating the constraint. Clearly a solution exists if and only if $\pi = (q, K, R)\in \Pi$ is such that $|\mu(K, R)|^2 < \varepsilon^2$. Solving in $J \in (0, +\infty)$ the equation
$\widetilde{\mathcal{M}}(J, \pi) = \varepsilon^2$ give the result.
\end{proof}

It then natural to find the optimized parameter $\pi$ as a solution to,
\begin{equation}
\label{eq:problem_c_tilde_post_J}
    \underset{ \substack{\pi =(q,K,R) \in \Pi \\ |\tilde{\mu}(K, R)| < \varepsilon}}{\min} \tilde{\mathcal{C}}_{\tau}(J(\varepsilon, \pi), \pi)
\end{equation}

\subsubsection{Optimization of $q$}

To solve (\ref{eq:problem_c_tilde_post_J}), building upon Lemaire-Pagès \cite{lemairePages17}, we introduce the notion of effort of an MLMC estimator alongside with its tractable upper bound under $(\mathrm{Var}_{\beta})$.
\begin{definition}
    For all $\pi \in \Pi$ we define $\phi_{\tau}(\pi)$ the effort of an MLMC estimator defined by $\phi_{\tau}(\pi) := v(\pi)\kappa_{\tau}(\pi)$
    and $\bar{\phi}_{\tau}(\pi)$ its upper bound under $(\mathrm{Var}_{\beta})$ defined by $\bar{\phi}_{\tau}(\pi) := \bar{v}(\pi)\kappa_{\tau}(\pi)$.
\end{definition}

Noticing that, for all $\varepsilon > 0$, $\pi \in \Pi$,
\begin{equation}
\label{eq:prop_J_cost}
    \tilde{\mathcal{C}}_{\tau}(J(\varepsilon, \pi), \pi) = J(\varepsilon, \pi) \kappa_{\tau}(\pi) = \frac{\overline{v}(\pi) \kappa_{\tau}(\pi)}{\varepsilon^2 - \tilde{\mu}^2(K, R)} = \frac{\overline{\phi}_{\tau}(\pi)}{\varepsilon^2 - \tilde{\mu}^2(K, R)}\,,
\end{equation}
therefore problem (\ref{eq:problem_c_tilde_post_J}) becomes
\begin{equation*}
    \underset{ \substack{\pi = (q,K,R) \in \Pi \\ |\tilde{\mu}(K, R)| < \varepsilon}}{\min} \frac{\overline{\phi}_{\tau}(\pi)}{\varepsilon^2 - \tilde{\mu}^2(K, R)}
\end{equation*}

Notice how the parameter $q$ only impact the numerator of the value function and not the denominator nor the constraint. Consequently, the optimal $q$ as a function of $K$ and $R$ is the solution of the problem
\begin{equation}
\label{eq:min_problem_q}
    \underset{\substack{q \in [0, 1]^{\mathbb{N}} \\ \forall r \in \{1, \dots, R\} \, q_r > 0\\
    \sum_{r = 1}^R q_r = 1}}{\min} \; \overline{\phi}_{\tau}(q, K, R) \,.
\end{equation}

For convenience, we denote $\Pi_0 = (0, +\infty) \times \mathbb{N}$ the space of MLMC parameters other than $J$ and $q$. Following Lemaire-Pagès \cite{lemairePages17}, to optimize $q$ we rely on their following Lemma.

\begin{lemma}[\cite{lemairePages17}, Lemma 3.5]
\label{lem:cs_optimization}
    Let $R \in \mathbb{N}$ and for all $j \in \{1, \dots, R\}$, let $a_j > 0$, $b_j > 0$ and $q_j > 0$ such that $\sum_{r = 1}^R q_r = 1$, then
    \begin{equation*}
        \left( \sum_{j = 1}^R \frac{a_j}{q_j} \right) \left( \sum_{j = 1}^R b_j q_j \right) \geq \left( \sum_{j = 1}^R \sqrt{a_j b_j} \right)^2
    \end{equation*}
    and equality holds if and only if $q_j = \frac{\sqrt{a_j b_j^{-1}}}{\mu}$ where $\mu = \sum_{r = 1}^R \sqrt{a_j b_j^{-1}}$
\end{lemma}

\begin{proposition}
\label{prop:optim_form_for_q}
    Let $\pi_0 = (K, R) \in \Pi_0$, then (\ref{eq:min_problem_q}) admits a unique solution $(q_r(\pi_0))_{r \in \{1, \dots, R\}}$ defined by
    \begin{equation}
    \label{eq:optim_q}
        \forall r \in \{1, \dots, R\}, \quad q_r(\pi_0) = \frac{\bar{\sigma}(r, K)}{\sqrt{\gamma_{\tau}(K_r)}\mu_{\pi_0}}
    \end{equation}
    where 
    \begin{equation*}
        \mu_{\pi_0} = \sum_{r = 1}^{R} \frac{\bar{\sigma}(r, K)}{\sqrt{\gamma_{\tau}(K_r)}}
    \end{equation*}
    is a normalizing constant such that $\sum_{r = 1}^R q_r(\pi_0) = 1$.
    Letting $\bar{\phi}^*_{\tau} : \Pi_0 \longrightarrow \mathbb{R}$ be defined as
    \begin{equation*}
        \forall \pi_0 \in \Pi_0, \quad \bar{\phi}^*_{\tau}(\pi_0) = \left(\sum_{r = 1}^R \bar{\sigma}(r, K) \sqrt{\gamma_{\tau}(K_r)} \right)^2
    \end{equation*}
    then for all $\pi_0 \in \Pi_0$, 
    \begin{equation}
    \label{eq:prop:q}
        \bar{\phi}_{\tau}(q(\pi_0), \pi_0) = \bar{\phi}^*_{\tau}(\pi_0)
    \end{equation}
\end{proposition}
\begin{proof}
    Let $(K, R) \in \Pi_0$. Notice that for $q \in [0,1]^{\mathbb{N}}$ such that for all $r \in \{1, \dots, R\}$, $q_r > 0$ and $\sum_{r = 1}^R q_r = 1$,
    \begin{equation*}
        \bar{\phi}_{\tau}(q, K, R) = \bar{v}(q, K, R) \kappa_{\tau}(q, K, R) = \left( \sum_{r = 1}^R \frac{\bar{\sigma}^2(r, K)}{q_r} \right) \left( \sum_{r = 1}^R q_r \gamma_{\tau}(K_r) \right) \,.
    \end{equation*}
    Then a direct application of Lemma \ref{lem:cs_optimization} with $a_j = \bar{\sigma}^2(r, K)$, $j \in \{1, \dots, R\}$ and $b_j = \gamma_{\tau}(K_j)$, $j \in \{1, \dots, R\}$ gives the result.
\end{proof}

\begin{remark}
    Notice that when $\tau = 0$ we recover the optimal parameters from Table \ref{tab:closed_mlmc_parameters_v2}.
\end{remark}

The optimal parameter $\pi_0 = (K,R) \in \Pi_0$ solving (\ref{eq:problem_c_tilde_post_J}) are then solution to the problem
\begin{equation}
\label{eq:optim_pb_R_K_both}
    \underset{\substack{\pi_0 \in \Pi_0 \\ |\tilde{\mu}(\pi_0)| < \varepsilon}}{\min} \frac{\bar{\phi}_{\tau}^*(\pi_0)}{\varepsilon^2 - \tilde{\mu}^2(\pi_0)}
\end{equation}

\subsubsection{Optimization of $K$ and $R$}

This optimization problem can be solved efficiently with a numerical optimization procedure. To do so we proceed in two steps. First fixing $R \in \mathbb{N}$ we solve,
\begin{equation}
\label{eq:pb_min_K_optim}
    \min_{\substack{K \in (0, +\infty) \\ |\tilde{\mu}(K, R)| < \varepsilon}} \frac{\bar{\phi}_{\tau}^*(K, R)}{\varepsilon^2 - \tilde{\mu}^2(K, R)} \,.
\end{equation}
Note that for $K \in (0, +\infty)$ in the standard MLMC case
\begin{equation*}
    |\tilde{\mu}(K, R)| < \varepsilon \Longleftrightarrow \underline{K}(\varepsilon, R) := \left \lfloor \frac{|c_1|^{\frac{1}{\alpha}}}{\varepsilon^{\frac{1}{\alpha}}2^{R-1}}  \right \rfloor + 1\leq \lceil K \rceil
\end{equation*}
while in the ML2R case
\begin{equation*}
    |\tilde{\mu}(K, R)| < \varepsilon \Longleftrightarrow \underline{K}(\varepsilon, R):=\left \lfloor \frac{\tilde{c}^{\frac{1}{\alpha}}}{\varepsilon^{\frac{1}{\alpha R}} 2^{\frac{R - 1}{2}}} \right \rfloor + 1 \leq \lceil K \rceil
\end{equation*}
thus the constraint is explicit and straightforward to implement in practice. Furthermore for all $K \in (0, + \infty)$, in the standard MLMC case, the value function can be written as
\begin{equation*}
    \frac{\bar{\phi}_{\tau}^*(K, R)}{\varepsilon^2 - \tilde{\mu}^2(K, R)}  = \frac{\left(\bar{\sigma}_1 \sqrt{\tau + \lceil K \rceil} + \sqrt{V_1} \lceil K \rceil^{\frac{-\beta}{2}}  \sum_{r = 2}^R \sqrt{\tau + \lceil K \rceil 2^{r-1}}2^{\frac{-\beta(r-1)}{2}} \right)^2}{\varepsilon^2 - \frac{c_1^2}{\lceil K \rceil^{2\alpha} 4^{\alpha(R-1)}}}
\end{equation*}
and in the ML2R case
\begin{equation*}
    \frac{\bar{\phi}_{\tau}^*(K, R)}{\varepsilon^2 - \tilde{\mu}^2(K, R)}  = \frac{\left(\bar{\sigma}_1 \sqrt{\tau + \lceil K \rceil} + \sqrt{V_1} \lceil K \rceil^{\frac{-\beta}{2}}  \sum_{r = 2}^R |W^R_r|\sqrt{\tau + \lceil K \rceil 2^{r-1}}2^{\frac{-\beta(r-1)}{2}} \right)^2}{\varepsilon^2 - \frac{\tilde{c}^{2R}}{\lceil K \rceil^{2\alpha R} 2^{\alpha R (R-1)}}} \,.
\end{equation*}
The value function is therefore fully explicit and easy to implement in practice. Moreover it is constant on each intervals $(n, n+1]$ for integers $n$ satisfying the constraint. Thus, the search for the optimal $K$ can therefore be restricted to integers values that satisfy the constraint. Accordingly we rewrite (\ref{eq:pb_min_K_optim}) as
\begin{equation}
\label{eq:pb_min_K_optim_entiers}
    \min_{K \in \{\underline{K}(\varepsilon, R), \underline{K}(\varepsilon, R) + 1, \dots  \}} \frac{\bar{\phi}_{\tau}^*(K, R)}{\varepsilon^2 - \tilde{\mu}^2(K, R)} \,.
\end{equation}
Since objective function diverges when $K$ increases, (\ref{eq:pb_min_K_optim_entiers}) is guaranteed to have a solution. In practice this problem is efficiently solved by standard optimization algorithms, as the objective function is typically unimodal. 

Let $K(\varepsilon, R)$ denotes a solution of (\ref{eq:pb_min_K_optim}). We then solve
\begin{equation}
\label{eq:pb_min_R_optim}
    \min_{R \in \{R_-(\varepsilon), R_-(\varepsilon)+1, \dots, R(\varepsilon)\}} \frac{\bar{\phi}^*_{\tau}(K(\varepsilon, R), R)}{\varepsilon^2 - \tilde{\mu}^2(K(\varepsilon, R), R)}
\end{equation}
where $R(\varepsilon)$ is the optimized $R$ found in Table \ref{tab:closed_mlmc_parameters_v2} and $R_{-}(\varepsilon)$ is such that $R_{-}(\varepsilon) \to +\infty$ as $\varepsilon \to 0$ and $1 \leq R_{-}(\varepsilon) \leq R(\varepsilon)$. For instance $R_{-}(\varepsilon) = \lceil \log_{10}(R(\varepsilon)) \rceil$ satisfy the conditions. This lower bound ensures the asymptotic guarantees stated in Theorem \ref{thm:new_optimization_resulting_cost} for the parameters. Nevertheless, in practical computations, the search for the optimal R can be performed over the range $1$ to $R(\varepsilon)$. This optimization problem is straightforward to solve, as the objective function can be evaluated efficiently.
\begin{remark}
\label{remark:in_practice_coef_cr}
    In practice, for the ML2R case, using $\tilde{c}^R$ as a proxy for $c_R$ in (\ref{eq:proxy_bias_ml2r}) is appropriate only when $R$ is large. However, since the optimization problems (\ref{eq:pb_min_K_optim}) and (\ref{eq:pb_min_R_optim}) also consider small values of $R$, a more accurate approximation of for $c_R$ is given by $c_1 a^R$ for some $a > 1$. The coefficient $c_1$ can typically be estimated reliably in a pre-processing phase. Accordingly, all instance of $\tilde{c}$ should be replace by $c_1^{\frac{1}{R}} a$ in the relevant expressions. Numerical experiments suggests that when dealing with indicator functions, choosing $a = 2$ or $3$ yields effective results. Note that in order to guarantee the satisfaction of the MSE constraint is satisfied, it is preferable to overestimate $c_R$ rather than to underestimate it. In our numerical experiments of Section \ref{sec:num_appli}, we will adopt this proxy with $a = 2$.
\end{remark}

\subsubsection{Complexity analysis for the new optimized parameters}
\label{sec:complexity_analysis}
In this section we formalize the new optimized parameters and prove  Theorem \ref{thm:new_optimization_resulting_cost}. Given a target precision $\varepsilon > 0$, we denote by $R^*(\varepsilon)$ the solution to the optimization problem~(\ref{eq:pb_min_R_optim}). Based on this value, we define $K^*(\varepsilon) := K^*(\varepsilon, R^*(\varepsilon))$, and we collect these into the vector $\pi^*_0(\varepsilon) := (K^*(\varepsilon), R^*(\varepsilon))$. Next, we define $q^*(\varepsilon) := q(\varepsilon, \pi^*_0(\varepsilon))$, and aggregate this with the previous parameters as $\pi^*(\varepsilon) := (q^*(\varepsilon), \pi_0(\varepsilon))$. Finally let $J^*(\varepsilon) := J(\varepsilon, \pi^*(\varepsilon))$. We summarize all the optimized parameters for the target precision $\varepsilon$ by the vector $\theta^*_{\varepsilon} := (J^*(\varepsilon), q^*(\varepsilon), K^*(\varepsilon), R^*(\varepsilon))$. For convenience, all these definitions are summarized in Table~\ref{tab:num_optim_mlmc_parameters}.

\begin{table}[H]
    \centering
    \renewcommand{\arraystretch}{3}
    \begin{tabular}{c | c | c}
        Parameter & Weighted (ML2R) & Classical (MLMC)\\
        \hline
        $R^*(\varepsilon)$ & \multicolumn{2}{|c}{$    \underset{{R \in \{R_-(\varepsilon), \dots, R(\varepsilon)\}} }{\arg\min} \; \frac{\bar{\phi}^*_{\tau}(K(\varepsilon, R), R)}{\varepsilon^2 - \tilde{\mu}^2(K(\varepsilon, R), R)}$} \\
        \hline
        $\underline{K}^*(\varepsilon)$& $\quad \quad \quad \quad \left \lfloor \frac{\tilde{c}^{\frac{1}{\alpha}}}{\varepsilon^{\frac{1}{\alpha R^*(\varepsilon)}} 2^{\frac{R^*(\varepsilon) - 1}{2}}} \right \rfloor + 1 \quad \quad \quad \quad$ & $ \left \lfloor \frac{|c_1|^{\frac{1}{\alpha}}}{\varepsilon^{\frac{1}{\alpha}}2^{R^*(\varepsilon)-1}}  \right \rfloor + 1$\\
        \hline
        $K^*(\varepsilon)$ & \multicolumn{2}{|c}{$    \underset{K \in \{\underline{K}^*(\varepsilon), \underline{K}^*(\varepsilon) + 1, \dots\}}{\arg\min} \frac{\bar{\phi}_{\tau}^*(K, R^*(\varepsilon))}{\varepsilon^2 - \tilde{\mu}^2(K, R^*(\varepsilon))}$} \\
        \hline
        $q^*(\varepsilon)$ & \multicolumn{2}{|c}{$q^*_1(\varepsilon) = \frac{\bar{\sigma}_1}{\sqrt{\tau + \lceil K^*(\varepsilon) \rceil}\mu^*_{\varepsilon}}\quad $ $\forall r = 2, \dots, R(\varepsilon) \quad q^*_r(\varepsilon) = \frac{\sqrt{V_1}|A_{r}^{R^*(\varepsilon)}|}{ \left \lceil K^*(\varepsilon) \right \rceil^{\frac{\beta }{2}}2^{\frac{\beta(r-1)}{2}} \sqrt{\tau + \left \lceil K^*(\varepsilon) \right \rceil 2^{r-1}}\mu^*_{\varepsilon}}$} \\
         &  \multicolumn{2}{|c}{with $\mu^*_{\varepsilon}$ such that $\sum_{r = 1}^{R^*(\varepsilon)} q^*_r(\varepsilon) = 1$} \\
         \hline
         $J(\varepsilon)$ & \multicolumn{2}{|c}{$\frac{\bar{v}(q^*(\varepsilon), K^*(\varepsilon), R^*(\varepsilon))}{\varepsilon^2 - \tilde{\mu}^2(K^*(\varepsilon), R^*(\varepsilon))}$}  \\
    \end{tabular}
    \vspace{0.5cm}
    \caption{New optimized parameters for the standard and weighted Multi-level estimators. Here $\tilde{c}$ is a constant such that $\tilde{c} > \tilde{c}_{\infty}$ in $(WE_{\alpha})$.}
    \label{tab:num_optim_mlmc_parameters}
\end{table}

Our objective is now to prove Theorem \ref{thm:new_optimization_resulting_cost}, which demonstrate a reduction in the computational cost when using $\theta^*_{\varepsilon}$ from Table \ref{tab:num_optim_mlmc_parameters} in place of $\theta_{\varepsilon}$ from Table \ref{tab:closed_mlmc_parameters_v2}, while preserving the same asymptotic MSE behavior. The proof of the theorem is based on technical Lemmas presented in Appendix \ref{apx:technical_lemmas_for_complexity_thm}.

\begin{theorem}
\label{thm:new_optimization_resulting_cost}
    Let $\varepsilon > 0$, $\tau \geq 0$, then $\tilde{\mathcal{C}}_{\tau}(\theta^*_{\varepsilon}) \leq \tilde{\mathcal{C}}_{\tau}(\theta_{\varepsilon})$. Assume that $(\mathrm{Var}_{\beta})$ holds. In the standard MLMC case assume furthermore that $(\mathrm{WE}_{\alpha, 1})$ holds while in the case of the ML2R assume that $(\mathrm{WE}_{\alpha})$ holds with $
    \underset{R \in \mathbb{N}}{\sup} \;\underset{k \in \mathbb{N}}{\sup} \; |\eta_R(k)| <+\infty$ instead. Then,
    \begin{equation*}
        \underset{\varepsilon \to 0}{\lim \sup} \; \varepsilon^{-2} \mathcal{M}(\theta^*_{\varepsilon}) \leq 1 \,.
    \end{equation*}
\end{theorem}
\begin{proof}
    Let $\varepsilon > 0$, $\tau \geq 0$ and recall that $J^*(\varepsilon) = J(\varepsilon, \pi^*(\varepsilon))$ and $q^*(\varepsilon) = q(K^*(\varepsilon), R^*(\varepsilon))$. Therefore by (\ref{eq:prop_J_cost}), (\ref{eq:prop:q}) and the definition of $\pi_0^*(\varepsilon)$ we have
    \begin{equation}
    \label{eq:proof_min_cost_1}
        \tilde{\mathcal{C}}_{\tau}(\theta^*_{\varepsilon}) = \frac{\bar{\phi}^*_{\tau}(K^*(\varepsilon), R^*(\varepsilon))}{\varepsilon^2 - \tilde{\mu}^2(K^*(\varepsilon), R^*(\varepsilon))} = \underset{{\substack{K \in (0, +\infty) \\ R \in \{R_{-}(\varepsilon), \dots, R(\varepsilon)\} \\ |\tilde{\mu}(K, R)| < \varepsilon} }}{\arg\min} \; \frac{\bar{\phi}^*_{\tau}(K,R)}{\varepsilon^2 - \tilde{\mu}^2(K, R)}
    \end{equation}
     Since from Lemma \ref{lemma:proxy_bias_satisfied_for_K_R}, $\pi_0(\varepsilon) = (K(\varepsilon), R(\varepsilon))$ is admissible for (\ref{eq:proof_min_cost_1}) we conclude that
     \begin{equation*}
         \tilde{\mathcal{C}}_{\tau}(\theta^*_{\varepsilon}) \leq \frac{\bar{\phi}^*_{\tau}(K(\varepsilon), R(\varepsilon))}{\varepsilon^2 - \tilde{\mu}^2(K(\varepsilon), R(\varepsilon))} \,.
     \end{equation*}

     Recall that $\bar{\phi}^*_{\tau}(\pi_0(\varepsilon)) = \bar{\phi}_{\tau}(q(\varepsilon, \pi_0(\varepsilon)), \pi_0(\varepsilon))$
     and by Proposition \ref{prop:optim_form_for_q},
     \begin{equation*}
         q(\varepsilon, \pi_0(\varepsilon)) = \underset{\substack{q \in [0, 1]^{\mathbb{N}} \\ \sum_{r = 1}^R q_r = 1}}{\arg\min} \; \bar{\phi}_{\tau}(q, \pi_0(\varepsilon)) \,.
     \end{equation*}
     Therefore,
     \begin{equation}
     \label{eq:proof_complexity_ineq_1}
          \tilde{\mathcal{C}}_{\tau}(\theta^*_{\varepsilon}) \leq  \frac{\bar{\phi}^*_{\tau}(K(\varepsilon), R(\varepsilon))}{\varepsilon^2 - \tilde{\mu}^2(K(\varepsilon), R(\varepsilon))} \leq \frac{\bar{\phi}_{\tau}(q(\varepsilon), \pi_0(\varepsilon))}{\varepsilon^2 - \tilde{\mu}^2(K(\varepsilon), R(\varepsilon))} \,.
     \end{equation}
     Recalling that for all $\pi = (q, K, R) \in \Pi$,
     \begin{equation*}
         \tilde{\mathcal{C}}_{\tau}(J(\varepsilon, \pi), \pi) = \frac{\bar{\phi}_{\tau}(\pi)}{\varepsilon^2 - \tilde{\mu}^2(K, R)} \,,
     \end{equation*}
     then we have
     \begin{equation}
    \label{eq:proof_complexity_ineq_2}
         \frac{\bar{\phi}_{\tau}(q(\varepsilon), \pi_0(\varepsilon))}{\varepsilon^2 - \tilde{\mu}^2(K(\varepsilon), R(\varepsilon))} = \tilde{\mathcal{C}}_{\tau}(J(\varepsilon, \pi(\varepsilon)), \pi(\varepsilon)) \,.
     \end{equation}
     where $\pi(\varepsilon) = (q(\varepsilon), \pi_0(\varepsilon))$. Since by definition of $J(\varepsilon, \pi(\varepsilon))$,
     \begin{equation*}
         J(\varepsilon, \pi(\varepsilon)) = \underset{\substack{J \in (0, +\infty) \\ \widetilde{\mathcal{M}}(J, \pi(\varepsilon)) \leq \varepsilon^2} }{\arg \min} \tilde{C}_{\tau}(J, \pi(\varepsilon))
     \end{equation*}
     and from Lemma \ref{lemma:m_tilde_satisfied_for_theta} $\widetilde{\mathcal{M}}(J(\varepsilon), \pi(\varepsilon)) \leq \varepsilon^2$ we have
     \begin{equation}
     \label{eq:proof_complexity_ineq_3}
         \tilde{\mathcal{C}}_{\tau}(J(\varepsilon, \pi(\varepsilon)), \pi(\varepsilon)) \leq \tilde{\mathcal{C}}_{\tau}(J(\varepsilon), \pi(\varepsilon)) = \tilde{\mathcal{C}}_{\tau}(\theta_{\varepsilon})
     \end{equation}
     Combining (\ref{eq:proof_complexity_ineq_1}), (\ref{eq:proof_complexity_ineq_2}) and (\ref{eq:proof_complexity_ineq_3}) then gives $\tilde{\mathcal{C}}_{\tau}(\theta^*_{\varepsilon}) \leq \tilde{\mathcal{C}}_{\tau}(\theta_{\varepsilon})$ which proves the first claim. \\

     We now prove the second claim. Denoting $\pi^*(\varepsilon) = (q^*(\varepsilon), K^*(\varepsilon), R^*(\varepsilon))$, recall from (\ref{def:mse_theta}),
     \begin{equation*}
         \mathcal{M}(\theta^*_{\varepsilon}) = \frac{v(\pi^*(\varepsilon))}{J^*(\varepsilon)} + \mu^2(K^*(\varepsilon), R^*(\varepsilon)) \,.
     \end{equation*}
     Under assumption $(\mathrm{Var}_{\beta})$, $v(\pi^*(\varepsilon)) \leq \bar{v}(\pi^*(\varepsilon))$. Therefore by definition of $J^*(\varepsilon)$,
     \begin{equation*}
         \mathcal{M}(\theta^*_{\varepsilon}) \leq \frac{\bar{v}(\pi^*(\varepsilon))}{J^*(\varepsilon)} + \mu^2(\pi_0^*(\varepsilon)) = \varepsilon^2 - \tilde{\mu}^2(\pi^*_0(\varepsilon)) + \mu^2(\pi^*_0(\varepsilon)) \,.
     \end{equation*}
     It suffice now to prove that
     \begin{equation*}
         \varepsilon^{-2} |\mu^2(\pi^*_0(\varepsilon)) - \tilde{\mu}^2(\pi^*_0(\varepsilon))| \underset{\varepsilon \to 0}{\longrightarrow} 0 \,.
     \end{equation*}
     We begin with the standard MLMC case. Under $(\mathrm{WE}_{\alpha, 1})$, from (\ref{eq:bias_mlmc}) and (\ref{eq:mu_tilde_mlmc}),
     \begin{equation*}
         |\mu^2(\pi^*_0(\varepsilon)) - \tilde{\mu}^2(\pi^*_0(\varepsilon))| = \tilde{\mu}^2(\pi^*_0(\varepsilon))  \left| \left(1 + \eta_1 \left(\frac{1}{K^*(\varepsilon) 2^{(R^*(\varepsilon) - 1) \alpha}} \right) \right)^2 - 1 \right|
     \end{equation*}
     Then since $\tilde{\mu}^2(\pi^*_0(\varepsilon)) \leq \varepsilon^2$ and $R^*(\varepsilon) \geq R_{-}(\varepsilon) \underset{\varepsilon \to 0}{\to} +\infty$ we have
     \begin{equation*}
         \varepsilon^{-2} \left|\mu^2(\pi_0^*(\varepsilon)) - \tilde{\mu}^2(\pi_0^*(\varepsilon)) \right| \leq \left| \left(1 + \eta_1\left(\frac{1}{K^*(\varepsilon) 2^{(R^*(\varepsilon) - 1) \alpha}} \right) \right)^2 - 1 \right| \underset{\varepsilon \to 0} {\longrightarrow} 0 \,.
     \end{equation*}
     which prove the result.

     We continue with the ML2R case. Under $(\mathrm{WE}_{\alpha})$  with $\underset{R \in \mathbb{N}}{\sup} \;\underset{k \in \mathbb{N}}{\sup} \; |\eta_R(k)| <+\infty$
     from (\ref{eq:bias_ml2r}) and (\ref{eq:proxy_bias_ml2r}),
     \begin{equation*}
         |\mu^2(\pi_0^*(\varepsilon)) - \tilde{\mu}^2(\pi_0^*(\varepsilon))| = \tilde{\mu}^2(\pi_0^*(\varepsilon)) \left| \left(\frac{c_{R(\varepsilon)}}{\tilde{c}^R} \right)^2 \left(\eta(\pi^*_0(\varepsilon)) + 1 \right)^2 - 1 \right| \,.
     \end{equation*}
    Since $\tilde{\mu}^2(\pi^*(\varepsilon)) \leq \varepsilon^2$,  $\eta$ is a bounded function, $\tilde{c} > \tilde{c}_{\infty}$ and $R^*(\varepsilon) \to +\infty$ as $\varepsilon \to 0$ we get
    \begin{equation*}
         \varepsilon^{-2} \left|\mu^2(\pi_0^*(\varepsilon)) - \tilde{\mu}^2(\pi_0^*(\varepsilon)) \right| \leq \left| \left(\frac{c_{R(\varepsilon)}}{\tilde{c}^R} \right)^2 \left(\eta(\pi^*_0(\varepsilon)) + 1) \right)^2 - 1 \right| \underset{\varepsilon \to 0}{\longrightarrow} 0 \,,
     \end{equation*}
     which complete the proof.
\end{proof}

\subsection{Influence of the parameter $\tau$ on the optimized parameters}
\label{sec:tau_analysis}
In this section we present a comprehensive analysis of the parameter $\tau$ and its effect on the optimized parameters of Table \ref{tab:num_optim_mlmc_parameters}. We recall that $\tau$ is a unitless parameter representing the ratio between outer sampling cost and inner sampling cost. Its value depends on the specific context, for example in the  context of insurance risk management, the outer sampling cost typically correspond to the cost of simulating the risk-factors and calibrating the Risk-Neutrals models to the realization. While the inner sampling cost typically represents the cost of simulating risk-neutral trajectories and computing the insurer profit and losses (P\&L) with an Asset Liability Model (ALM). When complex Risk-Neutral models are considered, the cost of their calibration may be substantial and therefore $\tau$ be large.

For $\tau \geq 0$, we let $g_{\tau}$ be the objective function of the optimization problem (\ref{eq:optim_pb_R_K_both}), namely
\begin{equation*}
    \forall \pi_0=(K, R) \in \Pi_0, \quad g_{\tau}(\pi_0) = \frac{\bar{\phi}_{\tau}^*(\pi_0)}{\varepsilon^2 - \tilde{\mu}^2(\pi_0)}
\end{equation*}

When $\tau$ grows large, a simple computation shows that for all $\pi_0 = (K, R) \in \Pi_0$ such that $|\tilde{\mu}(\pi_0)| < \varepsilon$
\begin{equation*}
    g_{\tau}(\pi_0) \underset{\tau \rightarrow +\infty}{\sim} \frac{\tau \left( \sum_{r = 1}^R \bar{\sigma}(r, K) \right)^2}{\varepsilon^2 - \tilde{\mu}^2(\pi_0)} \;.
\end{equation*}
This expression provides some insights on the role of $\tau$. As $\tau$ increases (particularity for large values), the cost penalty per unit of the "variance-like" term $\left( \sum_{r = 1}^R \bar{\sigma}(r, K) \right)^2$ also increases. In other words, a larger $\tau$ places greater emphasis on minimizing the variance across the levels. This often resulting in a smaller optimal number of levels $R$. At the same time, this tends to increase the optimal value for $K$ as the minimization shifts focus away from the cost per level. The outer sample repartition $q^*(\varepsilon)$ is also affected by $\tau$. Recalling that for all $r \in \{1, \dots, R^*(\varepsilon) \}$
\begin{equation*}
    q^*_{r}(\varepsilon) = \frac{\bar{\sigma}(r,K^*(\varepsilon))}{\sqrt{\tau + K^*(\varepsilon) 2^{r-1}} \mu^*_{\varepsilon}} \,,
\end{equation*}
straightforward computations shows that
\begin{equation*}
    q^*_{r}(\varepsilon) \underset{\tau \rightarrow +\infty}{\sim} \frac{\bar{\sigma}(r, K^*(\varepsilon))}{\sum_{i = 1}^{R^*(\varepsilon)} \bar{\sigma}(i, K^*(\varepsilon))} \; .
\end{equation*}
Therefore as $\tau$ goes to infinity the outer sample repartition boils down to a measure of the variance contribution of the considered level to the overall variance, with no consideration to its cost. This is expected as when $\tau$ gets large, the relative cost of each levels gets close to 1. This will also impact the total number of outer samples as a consequence. Recalling that,
\begin{equation*}
    J^*(\varepsilon) = \frac{\bar{v}(\pi^*(\varepsilon))}{\varepsilon^2 - \tilde{\mu}(\pi^*_0(\varepsilon))}
\end{equation*}
a straightforward computations shows
\begin{equation*}
    J^*(\varepsilon) \underset{\tau \to +\infty}{\sim} \frac{\left(\sum_{r = 1}^{R^*(\varepsilon)} \bar{\sigma}(r, K^*(\varepsilon)) \right)^2}{\varepsilon^2 - \tilde{\mu}(\pi^*_0(\varepsilon))} \,.
\end{equation*}
Overall when $\tau$ grows large, the control of the variance of each level will be prioritized over the control of its sampling cost leading generally to lower $R$ and larger $K$. In turns, this leads to less overall outer samples (lower $J$) allocated more to the higher levels. These effects will be demonstrated in the numerical experiments of Section \ref{sec:tau_num_exp}.

\section{Estimating probability of large losses}
\label{sec:nested_payoff_indicator}
In this section, we investigate the application of MLMC estimators to the estimation of the cumulative distribution function (c.d.f.) of $L$ at a fixed threshold $u$. Building on this, we also briefly discuss how MLMC can be used to estimate quantiles. The main point of focus is to discuss how the indicator function arising from pointwise c.d.f. estimation imposes major structural constraints on MLMC estimators.

\subsection{The indicator function framework}
\label{sec:indicator_case}
Defining a threshold $u \in \mathbb{R}$, we set $f := \mathbbm{1}_{. \leq u}$. In that case the target to estimate reads,
\begin{equation}
\label{eq:proba_large_loss_target}
    I = \mathbb{P}(L \leq u) \, = F_L(u).
\end{equation}
corresponding to a pointwise evaluation of the c.d.f. of $L$. Gordy-Juneja studied this particular problem in their seminal paper \cite{gordyJuneja10}, developing the theory of the nested Monte Carlo estimator. Other notable work in this context can be found in the litterature, such as Giles-Haji-Ali \cite{gilesAli19} who studied a stochastic version of the standard MLMC in this context based on ideas from Broadie et. al. \cite{broadiemoallemi2011}.

In this context, when the dimension of risk factors $X$ is $d = 1$, a result from Giorgi et. al. \cite{Giorgi_2020} (Proposition 5.1) gives theoretical conditions in order to fulfill the assumption $(\mathrm{WE}_{\alpha, R})$ for $R \geq 1$ and $\alpha = 1$. Based on this, for the rest of this section we will assume that $(\mathrm{WE}_{\alpha})$ holds with $\alpha = 1$. The indicator function $f$ being non-smooth, the control of the variance $(\text{Var}_{\beta})$ holds with only $\beta = \frac{1}{2}$ ; one can found theoretical conditions for this to hold in Giles-Haj-Ali \cite{gilesAli19} (Proposition 2.2) or Giorgi et. al. \cite{Giorgi_2020} (Proposition 5.2). Based on this, for the rest of this section we assume will that $(\mathrm{Var}_{\beta})$ holds with $\beta = \frac{1}{2}$. Looking back at Theorems \ref{thm:complexity_ML2R} and \ref{thm:complexity_mlmc}, note that the regime $\beta < 1$, is the regime in which the ML2R estimator is comparatively much more efficient than the standard MLMC.

The indicator function framework imply important structural constraint on the problem. In this case the variance at the first level $\bar{\sigma}^2(1, K)$ for $K \in (0, +\infty)$ can be developed under the assumption $(\mathrm{WE}_{1})$. Expanding on the work of Gordy-Juneja \cite{gordyJuneja10} we obtain the following.
\begin{proposition}
\label{prop:sigma_1_expansion}
    For all $K \in (0, +\infty)$ and all $R \in \mathbb{N}$,
    \begin{equation*}
        \bar{\sigma}^2(1, K) = I(1 - I) + (1 - 2I)\sum_{r = 1}^R \frac{c_r}{\lceil K \rceil ^r} - \sum_{r = 1}^{\lfloor \frac{R}{2} \rfloor} \frac{c_r}{\lceil K \rceil^{2r}} - \sum_{r = 1}^{R-1} \sum_{s = 1}^{\lfloor R - r \rfloor} \frac{c_r c_s}{\lceil K \rceil^{r + s}} + o \left( \frac{1}{\lceil K \rceil^R} \right)
    \end{equation*}
\end{proposition}
\begin{proof}
    The proof comes mainly from the fact that in the indicator framework $Y_{K_1}$ is a Bernoulli random variable therefore,
    \begin{equation*}
        \bar{\sigma}^2(1, K) = \mathrm{Var}[Y_{K_1}] = \mathbb{E}[Y_{K_1}](1 - \mathbb{E}[Y_{K_1}])
    \end{equation*}
    Then using the bias expansion from $(\mathrm{WE}_{1})$ we get the desired result after collecting the negligible terms.
\end{proof}
Since in this framework $Y_{K_1}$ is a Bernoulli random variable we always have $\bar{\sigma}^2(1, K) \leq \frac{1}{4}$, which gives automatically a bound $\bar{\sigma}_1^2$ in $(\mathrm{Var}_{\frac{1}{2}})$. However, in practice, this choice may be too loose. For example if $I = 0.5 \%$ (i.e $u$ is the 99.5 \% quantile of the conditional expectation) the principal term $I(1-I)$ will be close to $0.5\%$ which is 50 time less than $\frac{1}{4}$. For extreme threshold estimation in practice, our recommendation is to simply take $ \bar{\sigma}^2_1 := \tilde{I}(1 - \tilde{I})$ where $\tilde{I}$ is a rough a priori estimate of $I$. A supplementary correction term $(1 - 2\tilde{I})c_1$ could also be added to be more conservative.

Returning to Remark \ref{rem:mlmc_variance_indicator}, we emphasize that in the indicator function case, for $K \in \mathbb{N}$,
\begin{equation*}
    \mathrm{Var}[Y_{K}] \approx I(1 - I) \, ,
\end{equation*}
while the strong irregularity of $f$ yields $\mathrm{Var}[\Delta Y_{2K}]$ of the same order of magnitude. One can verify this phenomenon empirically by comparing $V_1$ in $(\mathrm{Var}_{\beta})$ with $I(1 - I)$, which are commonly found to be of the same order of magnitude (see Table \ref{tab:summary_structural_constants} in Section \ref{sec:num_appli}). This seems to not be the case in context where $f$ is smooth (see for example numerical section of \cite{lemairePages17}) where they observe values for $V_1$ much smaller than $\mathrm{Var}[Y_{K}]$.

This leads to the following implication : each additional level in the MLMC estimators results in a relatively significant variance increase. Therefore, for an added level to be worthwhile, one need either $K$ to be large enough to offset $V_1$ in $(\mathrm{Var}_{\frac{1}{2}})$ (which is challenging because of the slow $\frac{1}{2}$ rate of variance decay), or the bias reduction achieved must be large enough to compensate the increased variance. Accordingly, our numerical experiments reveal that in this context the ML2R estimator is much more efficient than the standard MLMC estimator, even in the non-asymptotic regime. We attribute this improvement to the strong increase in bias reduction per levels in the ML2R case due to the weights. A second observation is that the optimal number of levels $R$ are generally low (at most $2$ or $3$ for most practical computational budget), and typically the ML2R estimator only surpasses the efficiency of the traditional nested MC estimator for large computational budget (or equivalently for high required precisions).

\subsection{Nested Monte Carlo : closed-form optimized parameters}
\label{sec:nested_mc_closed_optim}
In the particular case of the standard nested MC estimator (where we consider $R^*(\varepsilon) = 1$), the optimized parameters $K^*(\varepsilon)$ in Table \ref{tab:num_optim_mlmc_parameters} admits an explicit expression.

\begin{proposition}
\label{prop:nested_explicit_K}
For all $\varepsilon > 0$ and all $\tau > 0 $ we define $K^+_{\tau}(\varepsilon)$ the solution of the optimization problem
\begin{equation}
\label{eq:optim_pb_for_K_1}
    \underset{\substack{K \in \mathbb{R}^*_+ \\ \frac{c_1^2}{K^{2\alpha}} < \varepsilon^2}}{\min} \; \frac{\bar{\sigma}_1^2(\tau + K)}{\varepsilon^2 - \frac{c_1^2}{K^{2\alpha}}}
\end{equation}
Then
\begin{equation*}
    \label{eq:optim_K_nested_tau}
        K^*_{\tau}(\varepsilon) = \begin{cases}
	\frac{|c_{1}|}{\varepsilon} 2\cos \left( \frac{1}{3} \arccos \left( \frac{\varepsilon \tau}{|c_{1,\eta}|} \right)\right) &, \; \varepsilon < \frac{|c_{1,\eta}|}{\tau} \\
	3 c_{1} &, \; \varepsilon = \frac{|c_{1,\eta}|}{\tau} \\
	\left( \frac{|c_{1,\eta}|}{\varepsilon} \right)^{\frac{2}{3}} \left[ \left( \tau + \left( \tau^2 - \frac{c_{1,\eta}^2}{\varepsilon^2} \right)^{\frac{1}{2}} \right)^{\frac{1}{3}} + \left( \tau - \left(  \tau^2 - \frac{c_{1,\eta}^2}{\varepsilon^2} \right)^{\frac{1}{2}} \right)^{\frac{1}{3}}\right]&, \; \varepsilon > \frac{|c_{1,\eta}|}{\tau}
    \end{cases}
\end{equation*}

Furthermore, observe that the optimization problem~(\ref{eq:optim_pb_for_K_1}) is, both in its objective function and its constraint, the continuous analogue of the discrete optimization problem~(\ref{eq:pb_min_K_optim_entiers}) with $R = 1$. Consequently, in this case, the optimal value $K^*(\varepsilon)$ is precisely either $\lceil K^+_{\tau}(\varepsilon) \rceil$ or $\lfloor K^+_{\tau}(\varepsilon) \rfloor$.
\end{proposition}
\begin{proof}
    See Appendix \ref{apx:proof_nested_explicit_K}.
\end{proof}

This results can be useful to bypass the numerical optimization procedure when restricting ourselves to the nested Monte Carlo estimator. Noting that $K^+_{\tau}(\varepsilon) \underset{\varepsilon \to +\infty}{\sim} K^*(\varepsilon)$, for $\varepsilon$ small enough, $K^+_{\tau}(\varepsilon)$ also provide information on the allocation for asymptotic values of $\tau$. Letting
\begin{equation*}
    \label{eq:optim_J_nested_tau}
    J^+_{\tau}(\varepsilon) := \frac{\bar{\sigma}^2_{1}}{\varepsilon^2 - \frac{c_1^2}{(K^+_{\tau}(\varepsilon))^2}}
\end{equation*}
be the associated optimized number of outer samples.  When $\tau$ goes to 0, a straightforward computation shows that for a fixed $\varepsilon > 0$
\begin{equation*}
     K^{+}_{\tau}(\varepsilon) \underset{\tau \rightarrow 0}{\longrightarrow} \frac{|c_1|\sqrt{3}}{\varepsilon}
\end{equation*}
and
\begin{equation*}
    J^{+}_{\tau}(\varepsilon) \underset{\tau \rightarrow 0}{\longrightarrow} \frac{\bar{\sigma}^2_1}{2 \varepsilon^2}
\end{equation*}
which corresponds to the usual allocation for the nested Monte Carlo when $\tau = 0$ (see for example Lemaire-Pagès \cite{lemairePages17}, Proposition 2.3). Similarly, letting $\tau$ going to infinity we see that
\begin{equation*}
    K^+_{\tau}(\varepsilon) \underset{\tau \rightarrow +\infty}{\sim} (2 \tau)^{\frac{1}{3}} \left(\frac{|c_1|}{\varepsilon} \right)^{\frac{2}{3}}
\end{equation*}
and
\begin{equation*}
    J^+_{\tau}(\varepsilon) \underset{\tau \rightarrow +\infty}{\sim} \frac{\bar{\sigma}^2_1}{\varepsilon^2}
\end{equation*}
which shows that $K^+_{\tau}(\varepsilon)$ goes to infinity as $\tau$ grows and that the proportion of the computational budget allocated to inner samples grows asymptotically in $\tau^{\frac{1}{3}}$.

\subsection{From probabilty to quantiles}
In this section, we briefly discuss how we can use the presented estimators to estimate a quantile of $L$. For all $\alpha \in (0,1)$ we let
\begin{equation*}
    q_{\alpha} = \inf \{ v \in \mathbb{R} : F_{L}(v) \geq \alpha \}
\end{equation*}
Then, when $L$ is a continuous random variable, $q_{\alpha}$ is the smallest solution to the equation
\begin{equation*}
    F_{L}(v) \geq \alpha
\end{equation*}
and furthermore if $F_{L}$ is strictly increasing, it is the unique solution to
\begin{equation}
\label{eq:quantile}
    F_{L}(v) = \alpha \,.
\end{equation}
For simplicity, we will assume that this is the case. A natural idea is to replace $F_{L}$ with an estimate $\hat{F}$ constructed via an MLMC estimator. For $\theta \in \Theta$ one can consider, with the notation of Definition \ref{def:general_mlmc_def},
\begin{equation*}
    \forall v \in \mathbb{R}, \quad \hat{F}_{\theta}(v) = \frac{1}{J_1} \sum_{j = 1}^{J_1} Y_{K_1}^{j}(v) + \sum_{r = 2}^{R} \frac{A^R_{r}}{J_r} \sum_{j = 1}^{J_r} \Delta Y_{K_r}^{j}(v)
\end{equation*}
where for all $v \in \mathbb{R}$
\begin{equation*}
    Y_{K_1}^{j}(v) = \mathbbm{1}_{\hat{E}^{j}_{K_1}(X) \leq v}
\end{equation*}
and
\begin{equation*}
    \Delta Y_{K_r}^{j}(v) = \mathbbm{1}_{\hat{E}^{f, j}_{K_r}(X) \leq v} - \frac{1}{2}(\mathbbm{1}_{\hat{E}^{c, j}_{K_r}(X) \leq v} + \mathbbm{1}_{\hat{E}^{c', j}_{K_r}(X) \leq v})
\end{equation*}
In other words, we fix the samples and let the threshold $v$ vary. In the particular case where $R = 1$, it is well known that (\ref{eq:quantile}) admits a closed-form solution, in that case, denoting $(\hat{E}_{K_1}^{(j)})_{j \in \{1, \dots, J_1\}}$ the order statistics of the sample $(\hat{E}_{K_1}^{j})_{j \in \{1, \dots, J_1\}}$ :
\begin{equation*}
    \hat{F}_{\theta}(v) = \hat{E}_{K_1}^{(j_{\alpha})}
\end{equation*}
where $j_{\alpha} = \lceil J \alpha \rceil$. However when $R \geq 2$ no such relation holds and we must numerically solve (\ref{eq:quantile}). We refer to Algorithm \ref{alg:simulatenous_estim_ml2r} for a practical implementation of this quantile estimator.

\section{Risk measurement in a Life-Insurance toy model}
In this section we apply MLMC estimators to a life insurance solvency monitoring context. The European Solvency II (SII) act require all European insurance companies to compute a Solvency Capital Requirement (SCR) which is defined as a 99.5\% quantile on their 1-year future own-fund loss $L_1 = OF_0 - OF_1$ where $OF_0$ is the current (deterministic) own-fund of the company and $OF_1$ is the (random) own-fund of the company in one year. SII requires insurance companies to evaluate their own-fund with a risk-neutral view and typically $OF_1 = \mathbb{E}_{\mathbb{Q}}[\sum_{t = 2}^T \tilde{F}_t | \mathcal{F}_1]$ where $\mathbb{Q}$ is a risk-neutral measure, $(\mathcal{F}_t)$ is a filtration representing market and actuarial informations and $(\tilde{F}_t)_{2 \leq t \leq T}$ are discounted future cash-flows for the company. Since an exact simulation of $OF_1$ is most often impossible insurance companies are in a typical nested MC framework where the conditional expectation must be approximated with a MC procedure.

As full assessment of the 99.5\% quantile of $L_1$ is computationally demanding, insurance company can rely on a simple "Standard-Formula" approach to the computation of the SCR. It uses predefined stress scenarios to various risk factors (such as market, credit, underwriting, and operational risks) to determine individual SCRs for each risk category. These SCRs are then aggregated into a total SCR using a specified correlation matrix, which accounts for diversification effects between different risk factors (see, e.g., \cite{Scherer2020TheSF}). However, insurers seeking a more accurate and tailored assessment of their risk profile are required to compute their SCR using their own "Internal Model" for the estimation of the 99.5\% quantile.

Our objective is to evaluate the complexities of the different MLMC estimators in the "Internal Model" context And will therefore focus on two metrics. The first is the evaluation of the c.d.f. of $L_1$ at a quantile level $q_{99.5\%}$ :
\begin{equation*}
    I = \mathbb{P}(L_1 \leq q_{99.5\%}) = 0.995
\end{equation*}
while the second is the evaluation of the quantile $q_{99.5\%}$ itself. The study of this framework was motivated by existing literature (Alfonsi et. al. \cite{alfonsiCherchaliAcevedo19}) that already successfully applied MLMC estimators in the context of the Standard Formula. To do so, we introduce a simplified life-insurance (ALM) model, mimicking standard contract characteristics. The simplicity of the model allow us to get closed formulas for $q_{99.5\%}$ for reference.

Our numerical experiments highlight the efficiency of the new MLMC parameters presented in Table \ref{tab:num_optim_mlmc_parameters} and the impact of $\tau$ on the calibrated parameters. We will also observe how ML2R ends up being much more efficient than MLMC and standard nested MC in estimating both type of targets.

\subsection{Market and mortality dynamics}
For life insurance contracts some of the main drivers of risks are financial market risks (stock, bond, credit \dots) and mortality risks (longevity, increased death rates, pandemics, \dots). In our simple model we will consider a financial market with a constant risk free rate $r \in \mathbb{R}$ and a stock index whose value is represented by the stochastic process $(S_t)_{t \geq 0}$. To simplify we consider a constant mortality rate $p \in [0, 1]$ for any individual (although the model could naturally be extended to actuarial mortality tables). Therefore in this model the only source of randomness is the evolution of the stock index. Accordingly, we define the information available to the insurance company as the filtration $(\mathcal{F}_t)_{t \geq 0}$ generated by the stochastic process $(S_t)_{t \geq 0}$. We assume that the stock process follows Black-Scholes dynamics with volatility $\sigma > 0$ and instantaneous return $\mu \in \mathbb{R}$. Namely letting $(W^{\mathbb{P}}_t)_{t \geq 0}$ be a Brownian motion under $\mathbb{P}$ we assume that
\begin{equation*}
    \forall t \geq 0, \; S_t = s_0 \exp \left( \left(\mu - \frac{\sigma^2}{2} \right)t + \sigma W^{\mathbb{P}}_t \right)
\end{equation*}
We denote $\mathbb{Q}$ the unique risk-neutral probability measure in this market and $(W^{\mathbb{Q}}_t)_{t \geq 0}$ the Brownian under $\mathbb{Q}$ such that
\begin{equation*}
    \forall t \geq 0, \; S_t = s_0 \exp \left( \left(r - \frac{\sigma^2}{2} \right)t + \sigma W^{\mathbb{Q}}_t \right)
\end{equation*}

\subsection{Description of the life-insurance contract}
We consider a life-insurance savings contract where policyholders make collectively an initial payment $MR_0 \in \mathbb{R}^*_+$ constituting the Mathematical Reserve of the company at time $t = 0$. The insurer invests this capital in the stock index $(S_t)_{t \geq 0}$ and offers the following financial guarantees~:
\begin{itemize}
    \item Minimum guaranteed rate $r_g \in \mathbb{R}$ : This is the minimum rate of appreciation on the savings of the policyholders
    \item Profit-sharing rate $\gamma \in [0,1]$ : Under french regulation $\gamma = 85\%$, this corresponds to the percentage of profits on the stock investments that must be redistributed to policyholders
\end{itemize}
The savings of the policyholders and any current additional interest (from minimum guaranteed rate or profit-sharing mechanism) are paid once at termination of the contract. The termination is either at a fixed maturity $T > 0$ of the contract or earlier if the policyholder dies. We assume that the contract is monitored yearly, and payment are made at the end of each year $t = 1, \dots, T$. \\

At inception time $t = 0$, the insurer invests $MR_0$ in the stock index which gets him $\phi_0 = \frac{MR_0}{s_0}$ shares of the stock. Then each year $1 \leq t \leq T - 1$, the following steps are carried :
\begin{enumerate}
    \item The value of the policyholders' savings $MR_{t-1}$ are appreciated at a rate
    \begin{equation*}
        r_s(t) :=  \max(r_g, \gamma \ln(R_t)) = r_g + (\gamma \ln(R_t) - r_g)^+
    \end{equation*}
    where $R_t = \frac{S_t}{S_{t-1}}$, $t \in \{1, \dots, T\}$, satisfying the financial guarantees of the minimum guaranteed rate and the profit sharing mechanism. The appreciated value is then
    \begin{equation}
    \label{eq:alm_update_mr_tilde}
        \widetilde{MR}_t := MR_{t-1} (1 + r_s(t))
    \end{equation}
    
    \item A proportion $p$ of policyholders dies during the year, triggering a termination of their contract. The insurer therefore must pay
    \begin{equation*}
        \widetilde{MR}_t p
    \end{equation*}
    to the deceased policyholders. This requires to sell
    \begin{equation*}
        \Delta \phi_t := \frac{-\widetilde{MR}_t p}{S_t}
    \end{equation*}
    share of the stock to provide the liquidity. The insurance company is left with $\phi_t := \phi_{t-1} + \Delta \phi_t$ shares of the stock.

    \item The value of the remaining policyholders' savings is then updated by 
    \begin{equation}
    \label{eq:alm_update_mr}
        MR_t := \widetilde{MR}_{t} (1 - p)
    \end{equation}
\end{enumerate}
At termination date $t = T$, the same steps are carried but with $p = 1$ (as all policyholders terminates their contract regardless if they died or not). The shareholders of the insurance company are then left with $\phi_T S_T \in \mathbb{R}$ from the liquidation of the remaining asset shares.

\subsection{Solvency monitoring of the portfolio}
An insurance company is said to be solvent if its Own Funds are positive ; here the OF is defined as the discounted risk-neutral expectation of the future cash-flows generated by the contract conditionally to current information :
\begin{equation*}
    \forall 0 \leq t \leq T, \quad OF_t := \mathbb{E}_{\mathbb{Q}}[e^{-r(T-t)} \phi_T S_T | \mathcal{F}_t]
\end{equation*}
The OF in the Solvency II framework represents the difference between the market value of the assets of the insurance company and the value a rationale agent would consent to bear the responsibility of the liabilities of the company.

At time $t = 0, \dots, T$, the SCR can then be defined as a $99.5\%$ quantile on the distribution of $-OF_{t+1}$ conditionally to $\mathcal{F}_t$ :
\begin{equation*}
    SCR_t := OF_{t} + \inf \{ x \in \mathbb{R} : \mathbb{P}(-OF_{t+1} \leq x | \mathcal{F}_t) \geq 99.5\% \} = q_{99.5\%}(L_t | \mathcal{F}_t)
\end{equation*}
where $L_t = OF_t - OF_{t+1}$ is the loss in Own Funds over one year in the future. The following proposition show that, in our model, $OF_t$ can be explicitly computed based on the state variables $\phi_t$, $MR_t$, $S_t$, as well as on the trajectory $(s_0, S_1, \dots, S_t)$. The proof of this proposition is based on technical Lemmas presented in Appendix \ref{apx:technical_lemma_of}. For convenience we first introduce the following notations :
\begin{definition}
    Let $(d_u)_{u \in\{1, \dots, T\}}$ be the sequence valued in $\mathbb{R}$ defined by,
    \begin{equation*}
        d_u := \begin{cases}
            p & u \in \{1, \dots, T-1\} \\
            1 & u = T
        \end{cases}
    \end{equation*}
    and call $(d_u)_{u \in\{1, \dots, T\}}$ the sequence of exit rates. Further we define, $z := 1 + r_g + \gamma \sigma (\Phi^{'}(d) + d \Phi(d))$ where $d := \frac{r - \frac{\sigma^2}{2} - \frac{r_g}{\gamma}}{\sigma} $.
\end{definition}

\begin{proposition}
\label{prop:closed_formula_neg_OF}
    Let $t \in \{0, \dots, T\}$,
    \begin{equation*}
        OF_t = \mathbb{E}_{\mathbb{Q}}[e^{-r(T-t)} \phi_T S_T |\mathcal{F}_t] = \phi_t S_t - MR_t \left[p \sum_{u = t}^{T - 1} e^{-r(u-t)} (1 - p)^{u - t - 1} z^{u - t} + e^{-r(T-t)} (1 - p)^{T - t - 1} z^{T - t} \right]
    \end{equation*}
    In particular, letting $\psi_t$ be the function from $\mathbb{R}^{t+1}$ to $\mathbb{R}$ defined by,
    \begin{equation*}
        \forall x = (x_0, x_1, \dots, x_t) \in (0, +\infty)^{t+1}, \; \psi_t(x) = g_t(x) x_t-  f_t(x) \left[p \sum_{u = t}^{T - 1} e^{-r(u-t)} (1 - p)^{u - t - 1} z^{u - t} + e^{-r(T-t)} (1 - p)^{T - t - 1} z^{T - t} \right] \,,
    \end{equation*}
    where 
    \begin{equation*}
        \forall x = (x_0, \dots, x_t) \in (0, +\infty)^{t+1}, \quad f_t(x) = MR_0 \prod_{u = 1}^{t} (1 - d_u) (1 + \rho(x_{u-1}, x_u)) \,,
    \end{equation*}
    \begin{equation*}
        \forall x = (x_0, x_1, \dots, x_t) \in (0, +\infty)^{t+1}, \quad g_t(x) = \phi_0 - MR_0 \sum_{i = 1}^{t} \frac{d_i}{x_i} \prod_{j = u + 1}^{i-1} (1 - d_j) \prod_{j = u + 1}^i (1 + \rho(x_{j-1}, x_j)) \,,
    \end{equation*}
    and
    \begin{equation*}
        \forall (x, x') \in (0, +\infty)^{2}, \quad \rho(x, x') = \max\left(r_g, \gamma\ln\left(\frac{x'}{x} \right) \right) \,,
    \end{equation*}
    then $\mathbb{E}_{\mathbb{Q}}[e^{-r(T-t)} \phi_T S_T |\mathcal{F}_t] = \psi_t(s_0, S_1, \dots, S_t)$ \,.
\end{proposition}
\begin{proof}
    For all $t \in \{0, \dots, T\}$, from Lemma \ref{lemma:rec_relation_phi_alm},
    \begin{equation*}
        \phi_T S_t = \phi_t S_T - MR_t \sum_{i = t + 1}^{T} \frac{d_i S_T}{S_i} \prod_{j = t+1}^{i-1} (1 - d_j) \prod_{j = t+1}^{i} ( 1 + \rho(S_{j-1}, S_j))
    \end{equation*}
    Now since $\phi_t$ and $MR_t$ are $\mathcal{F}_t$-measurable,
    \begin{equation*}
        \mathbb{E}_{\mathbb{Q}}[\phi_T S_T | \mathcal{F}_t] = \phi_t \mathbb{E}_{\mathbb{Q}}[S_T | \mathcal{F}_t] - MR_t \sum_{i = t+1}^{T} d_i \prod_{j = t+1}^{i-1} (1 - d_j) \mathbb{E}_{\mathbb{Q}} \left[\frac{S_T}{S_i} \prod_{j = t+1}^{i} (1 + \rho(S_{j-1}, S_j)) | \mathcal{F}_t \right]
    \end{equation*}
    A well-known fact in the Black-Scholes model is that
    \begin{equation*}
        \mathbb{E}_{\mathbb{Q}}[S_T | \mathcal{F}_t] = S_t e^{r(T - t)} \,.
    \end{equation*}
    Furthermore, for all $i \in \{t+1, \dots, T\}$, $\frac{S_t}{S_i} \prod_{j = t+1}^{i} (1 + \rho(S_{j-1}, S_j))$ is independent from $\mathcal{F}_t$, therefore
    \begin{align*}
        \mathbb{E}_{\mathbb{Q}} \left[\frac{S_T}{S_i} \prod_{j = t+1}^{i} (1 + \rho(S_{j-1}, S_j)) | \mathcal{F}_t \right] &= \mathbb{E}_{\mathbb{Q}} \left[\frac{S_T}{S_i} \prod_{j = t+1}^{i} (1 + \rho(S_{j-1}, S_j)) \right] \\
        &= \mathbb{E}_{\mathbb{Q}} \left[\frac{S_T}{S_i} \right] \prod_{j = t+1}^{i} \mathbb{E}[1 + \rho(S_{j-1}, S_j)] \\
        &= e^{T-i} z^{i - t}
    \end{align*}
    where the second equality comes from the independence of the random variable, and the last from Lemma \ref{lemma:z_formula_alm}. In turns we get
    \begin{equation*}
        \mathbb{E}_{\mathbb{Q}}[e^{-r(T-t)}\phi_T S_T | \mathcal{F}_t] = \phi_t S_t - MR_t \sum_{i = t+1}^{T} d_i e^{(i-t)} z^{i-t} \prod_{j = t+1}^{i-1} (1 - d_j)
    \end{equation*}
    Using that $d_i = p$ for $i \in \{0, \dots, T-1\}$ and $d_T = 1$ then give the first claim. The second claim is a direct application of the first claim with Lemma \ref{lemma:mr_rec_formula_alm} and Lemma \ref{lemma:rec_relation_phi_alm}.
\end{proof}

From now on, to simplify we will set ourselves in the case $t = 0$. Notice then that $L_1 = \psi_0(s_0) - \psi_1(s_0, S_1)$ and recalling that $s_0$ is deterministic we let $\psi$ be the (deterministic) function from $(0, +\infty)$ to $\mathbb{R}$ defined by,
    \begin{equation*}
        \forall x \in (0, + \infty), \quad \psi(x) := \psi_0(s_0) - \psi_1(s_0, x) \,,
    \end{equation*}
    then $\psi$ is such that $L_1 = \psi(S_1) \; a.s$. \\
    
    To relate this framework to the context of the notation of the nested MC framework used in previous sections, let $X$ be a random variable whose distribution under $\mathbb{P}$ is the same as that of $S_1$ under $\mathbb{P}$. Define $U = (U_1, \dots, U_{T-1})$ as a random vector, independent of $X$, whose distribution under $\mathbb{P}$ matches that of the independent Gaussian increments $(W_{2}^{\mathbb{Q}} - W^{\mathbb{Q}}_{1}, \dots, W_{T}^{\mathbb{Q}} - W_{T-1}^{\mathbb{Q}})$ under $\mathbb{Q}$. For all $t > 0$, we let $\phi_{x, t}$ be the functional mapping $t$ Gaussian increments to the value of the stock at time $t$ in the Risk-Neutral Black-Scholes model, starting from $x$ at time $0$. Now, define the function $F$ for $x \in (0, +\infty)$, $u = (u_1, \dots u_{T-1}) \in (0, +\infty)^{T-1}$ by
    \begin{equation*}
        F(x, u) = \psi_0(s_0) - e^{-r(T-t)} g_T(s_0, x, \phi_{x, 1}(u_1), \dots, \phi_{x, T-1}(u_1, \dots, u_{T-1})) \phi_{x, T-1}(u_1, \dots, u_{T-1})
    \end{equation*}
    where $g_T$ (as defined in \ref{eq:phi_eq_g}) maps stock price trajectories to the terminal number of stock shares. With these definitions, we have $L_1 = \mathbb{E}[F(X, U) | X]$ where $X$ and $U$ are independent. The propositions below give us a closed formula for $SCR_0$ that we will use as reference value for our numerical experiments.
\begin{proposition}
\label{prop:decreasing_condition_psi}
    The function $\psi$ is decreasing on $(0, x_1]$ and $[x_2, +\infty)$ where
    \begin{equation*}
    \begin{cases}
        x_1 := s_0 e^{\frac{r_g}{\gamma}} \\
        x_2 := s_0 \gamma \left((1 - p) \left[p\sum_{t = 1}^{T-2} e^{-rt} (1 - p)^{t - 1} z^t + e^{-r(T-1)}(1-p)^{T-2} z^{T-1} \right] + p \right)
    \end{cases}
    \end{equation*}
    In particular if $x_1 \geq x_2$ the value function $\psi$ is non-increasing on $\mathbb{R}^*_+$.
\end{proposition}
\begin{proof}
    A straightforward calculation shows that $\psi$ is differentiable on $(0, x_1)$ and $(x_1, +\infty)$ with
    \begin{equation*}
        \psi'(x) = \begin{cases}
            - \phi_0 & x < x_1 \\
            MR_0\frac{\gamma}{x} \left[p\sum_{t = 1}^{T-2} e^{-rt} (1 - p)^{t - 1} z^t + e^{-r(T-1)}(1-p)^{T-2} z^{T-1}  \right] - \phi_0 & x_1 < x \\
        \end{cases}
    \end{equation*}
    Therefore when $x < x_1$, and when $x_2 \geq x$ we have $\psi'(x) \leq 0$, noting that $\psi$ is continuous at $x_1$ we get the result.
\end{proof}

\begin{proposition}
    With the notations of Proposition \ref{prop:decreasing_condition_psi}, if $x_1 \geq x_2$ we have
    \begin{equation*}
        SCR_0 = q_{1 - \alpha}(\psi(S_1)) = \psi \left(s_0 \exp \left(\mu - \frac{\sigma^2}{2} + \sigma \Phi^{-1}(\alpha) \right) \right)
    \end{equation*}
    with $\alpha = 0.5\%$
\end{proposition}
\begin{proof}
    This is an immediate consequence of the non-increasing property of $x \mapsto \psi(x)$ and the non-decreasing property of $u \mapsto s_0 \exp \left(\mu - \frac{\sigma^2}{2} + \sigma u \right)$.
\end{proof}
In our numerical experiments, we will use the parameters reported in Table \ref{tab:params_alm} which yields $q_{99.5\%} \approx 252.76$.

In Figure \ref{fig:hist_and_psi_x} a) we show an histogram of exact sampling of $L_1$ and in Figure \ref{fig:hist_and_psi_x} b) we represent $\psi$. We observe that the loss grows as the stock perform poorly during the first year, with a comparatively steep increase when going below the 100 mark. This result in no loss in own-fund for most of scenarios and a comparatively severe loss in extreme downward stock scenarios.

\begin{figure}[H]
    \centering
    \subfloat[\centering Density estimation of $L_1$]{{\includegraphics[width=.4 \textwidth]{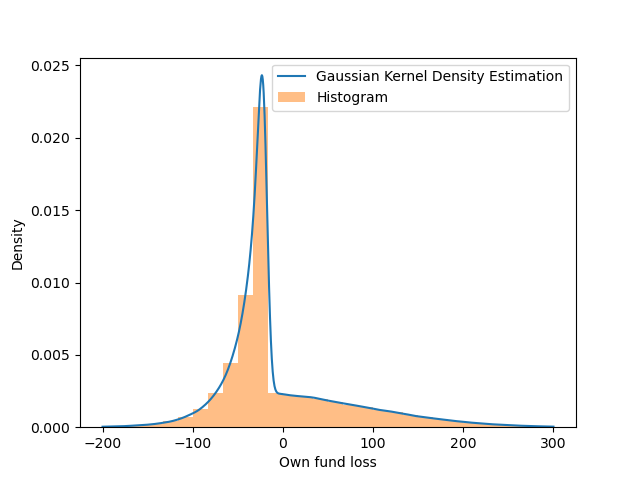} }}
    \qquad
    \subfloat[\centering $\psi(x)$ as function of $x$]{{\includegraphics[width=.4 \textwidth]{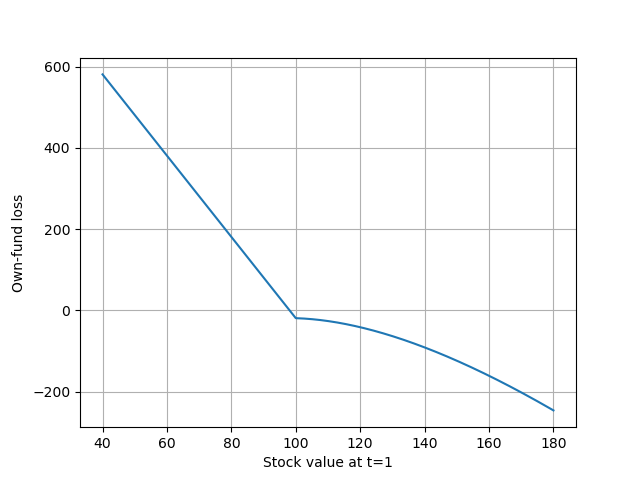} }}
    \caption{Characteristics of the own-fund loss}
    \label{fig:hist_and_psi_x}
\end{figure}

\section{Numerical experiments}
\label{sec:num_appli}
\subsection{Structural constants in the framework}
In this section we conduct a pre-processing where we estimate the structural constants of the framework. Here we do not restrict ourselves on the simulation budget used to estimate these constants. In practice only a fraction of the total computational budget should be used.

\subsubsection{Bias structural constants}
Following $(\mathrm{WE}_{1})$, we begin with the the structural constant $c_1$, based on the relation~:
\begin{equation*}
    \forall K \in \mathbb{N}, \quad \mathbb{E}[\Delta Y^A_{2K}] = \mathbb{E}[Y_{2K} - Y_K] \approx -\frac{c_1}{2K}
\end{equation*}
and with the structural constant $c_2$ based on the relation :
\begin{equation*}
    \forall K \in \mathbb{N}, \quad \mathbb{E}[2Y_{4K} - 3Y_{2K} + Y_K] \approx \frac{6 c_2}{(4K)^2}
\end{equation*}
In our setting choosing $c_1 = 0.025$ and $c_2 = 0.05$ seems empirically appropriate (see Figure \ref{fig:structural_const_plots} a) and b) ). Going back to Remark \ref{remark:in_practice_coef_cr}, we make the assumption that the sequence of coefficient $(c_r)_{r \in \mathbb{N}}$ follows the parametric form $c_r = c_1 a^{r - 1}$ where we choose $a = \frac{c_2}{c_1} = 2$. In particular this leads to $\tilde{c} = a = 2$. Notice however that $c_2$ is very hard to estimate (as shown by the large confidence interval in Figure \ref{fig:structural_const_plots} b) ) and is mostly infeasible in practice, therefore one must often blindly choose $a$. Our recommendation is to take typically $a$ to be $2$ or $3$.

\subsubsection{Variance structural constants}
The structural constant $\bar{\sigma}_1^2$ is choose based on Proposition \ref{prop:sigma_1_expansion}, namely we consider $\bar{\sigma}_1^2 \approx I(1-I) \approx 0.5\%$. Following $(\mathrm{Var}_{1/2})$ the structural constant $V_1$ is based on the relation:
\begin{equation*}
    \forall K \in \mathbb{N}, \quad \text{Var}[\Delta Y_{2K}] \approx \frac{V_1}{\sqrt{2K}}
\end{equation*}
Note that this constant depend on whether we consider the antithetic version of the level, $\Delta Y_{2K}^{A}$ or the non-antithetic version $\Delta Y^{S}_{K_r}$. In Figure \ref{fig:structural_const_plots} c) we plot the estimated relation in the framework for both the antithetic level and non antithetic level. In the end we see that $V^{A}_1 = 0.01$ with antithetic sampling and $V^{S}_1 = 0.02$ without are appropriate. As an illustration of Theorem \ref{thm:systematic_antithetic}, the antithetic sampling divide roughly by 2 the variance on each levels despite the fact that we are using indicator function payoff. 

In Table \ref{tab:summary_structural_constants} we summarize the structural parameters used in the experimentation. Notice how $\bar{\sigma}_1^{2}$ and $V_{1}$ are of the same order of magnitude, this is a common phenomenon in indicator function $f$ framework. This entail that each additional levels in MLMC estimators are (comparatively with the variance on the first level) adding a significant amount of variance. Therefore qualitatively, we expect to use a rather low number of levels.

\begin{remark}
    Since the structural constants are estimated using $f(x) = \mathbbm{1}_{x \leq q_{99.5\%}}$, if we are interested in estimating the quantile in practice a rough a-priori value for $q_{99.5\%}$ is required in order to estimate the strctural constants.
\end{remark}

\begingroup
\setlength{\tabcolsep}{10pt}
\renewcommand{\arraystretch}{1.5}
\begin{table}[H]
\centering
 \begin{tabular}{c | c} 
 \hline
 Parameter & Value \\
 \hline
 Risk free rate $r$ & $5\%$  \\
 Stock vol $\sigma$ & $15\%$ \\
 Stock drift $\mu$ & $8\%$ \\
 Stock initial value $s_0$ & $100$ \\
 \hline
 Horizon $T$ & $10$ years \\
 Minimum guaranteed rate $r_g$ & $0 \%$ \\
 Profit-sharing rate $\gamma$ & $85\%$ \\
 Death rate $p$ & $2\%$ \\
 Initial Mathematical Reserve $MR_0$ & 1000 \\
 \hline
 \end{tabular}
\vspace{0.5cm}
\caption{Parameters for the numerical experiments}
\label{tab:params_alm}
\end{table}
\endgroup

\begin{figure}
    \centering
    \subfloat[\centering First order structural bias plot. In blue the empirical estimation of $|\mathbb{E} (Y_{2K} - Y_K)|$. In faded blue the 95\% confidence interval. In orange the proxy relation for $c_1 = 0.025$.]{{\includegraphics[width=.45 \textwidth]{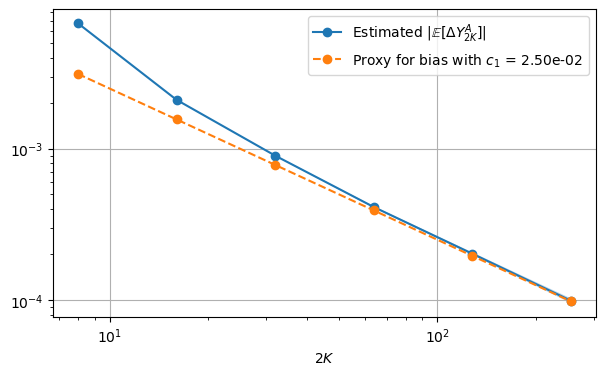} }}
    \qquad
    \subfloat[\centering Second order structural bias plot. In blue the empirical estimation of $|\mathbb{E} (2Y_{4K} -3 Y_{2K} + Y_K)|$. In faded blue the 95\% confidence interval. In orange the proxy relation for $c_2 = 0.05$.]{{\includegraphics[width=.45 \textwidth]{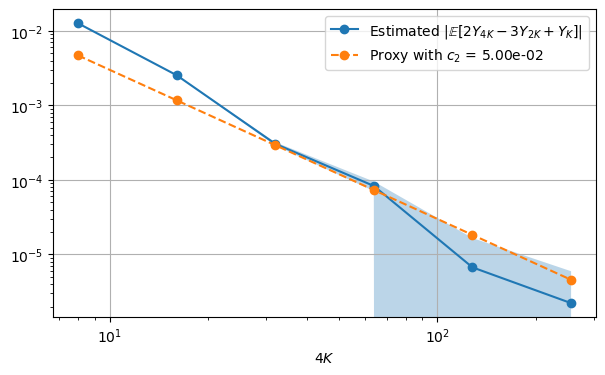} }}
    \subfloat[\centering Variance structural bias plot. In solid blue the empirical estimation of $\mathrm{Var}(\Delta Y^{A}_{2K})$ and in dashed blue the proxy relation for $V_1^{A} = 0.01$. In solid orange the empirical estimation of $\mathrm{Var}(\Delta Y^{S}_{2K})$ and in dashed orange the proxy relation for $V_1^{S} = 0.02$.]{{\includegraphics[width=.45 \textwidth]{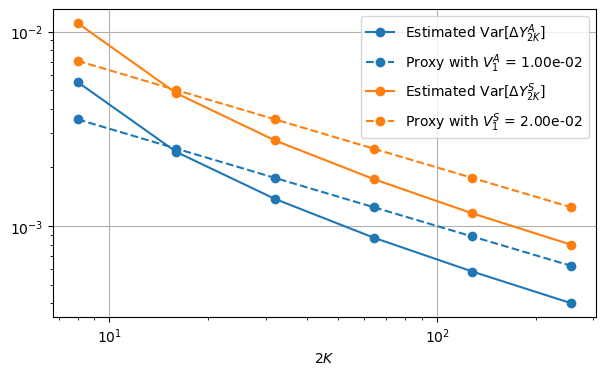} }}
    \caption{Structural constants plots.}
    \label{fig:structural_const_plots}
\end{figure}

\begingroup
\setlength{\tabcolsep}{10pt}
\renewcommand{\arraystretch}{1.5}
\begin{table}[H]
    \centering
    \begin{tabular}{c | c }
        \hline
        Structural constant & Value \\
        \hline
        First bias coefficient $c_1$ & $0.025$\\
        Bias coefficient geometric growth $a$ & $2$\\
        Variance decay constant with antithetic sampling $V^{A}_1$ & $0.010$ \\
        Variance decay constant without antithetic sampling $V^{C}_1$ & $0.020$\\
        Variance at the first level $\bar{\sigma}_1^2$ & $0.5\%$ \\
        \hline
    \end{tabular}
    \vspace{0.5cm}
    \caption{Summary of the structural constants in the framework.}
    \label{tab:summary_structural_constants}
\end{table}
\endgroup

\subsection{Methodology and Benchmarking}
\label{sec:num_methodo_benchmarking}
For all prescribed precision $\varepsilon > 0$, we consider 5 estimators :
\begin{enumerate}
    \item the standard nested Monte Carlo estimator that consist of using the parameters in Table \ref{tab:num_optim_mlmc_parameters} where $R^*(\varepsilon) = 1$ is forced regardless of $\varepsilon$
    \item the weighted Multi-level Monte Carlo (ML2R) estimator with the parameters of Table \ref{tab:num_optim_mlmc_parameters} left column
    \item the classical Multi-level Monte Carlo (MLC) estimator with the parameters of Table \ref{tab:num_optim_mlmc_parameters} right column
    \item the ML2R estimator with the standard parameters of Table \ref{tab:closed_mlmc_parameters_v2} left column
    \item the MLMC estimator with the standard parameters of Table \ref{tab:closed_mlmc_parameters_v2} right column
\end{enumerate}

The methodology is the following, we consider a sequence of target precision $(\varepsilon_n)_{n \in \{1, \dots, N\}}$. For each estimator $\hat{I}$ and each $n \in \{1, \dots, N\}$ we compute the corresponding parameters $\theta_n = \theta_{\varepsilon_n} \in \Theta$. The computational cost associated with each set of parameters is then evaluated.

For each set of parameter $\theta_n$, we perform $M \in \mathbb{N}$ independent estimation of $I$ and $q_{99.5\%}$ by generating samples according to the specified parameters. Using these estimations, we calculate the empirical RMSE. Importantly, for each sample, we simultaneously estimate both $I$ and $q_{99.5\%}$ from the same sample. To illustrate we report in Algorithm \ref{alg:simulatenous_estim_ml2r}, the pseudo-code for the simultaneous estimation in the ML2R case with antithetic sampling. Although the root-finding step in Algorithm \ref{alg:simulatenous_estim_ml2r} incurs additional computational cost that increases with $J$, we disregard this expense since, in practical applications, it would be negligible compared to the cost of the sampling step.

\begin{algorithm}
\caption{Estimate simultaneously $I$ and $q_{99.5\%}$ with an antithetic ML2R.}
\begin{algorithmic} 
\Require $\theta = (J, q, K, R) \in \Theta$, $u$ the c.d.f. evaluation point, $p$ the quantile level
\State $K_1 \gets \lceil K \rceil$
\State $J_1 \gets \lceil J q_1 \rceil$
\State $E_{K_1}$ be an i.i.d sample of $\hat{E}_{K_1}(X)$ with size $J_1$
\For{$r = 2, ..., R$}
\State $K_r \gets \lceil K \rceil 2^{r-1}$
\State $J_r \gets \lceil J q_r \rceil$
\State $(E^f_{K_r}, E_{K_r}^{c}, E_{K_r}^{c'})$ be an i.i.d sample of $(\hat{E}^f_{K_r}(X), \hat{E}^{c}_{K_r}(X), \hat{E}^{c'}_{K_r}(X))$ with size $J_r$
\EndFor
\Procedure{Evaluate}{v}
\State $\hat{I} \gets \frac{1}{J_1} \sum_{j = 1}^{J_1} \mathbbm{1}_{E_{K_1}[j] \leq v}$
\For{$r = 2, ..., R$}
\State $\hat{I} \gets \hat{I} + \frac{W_r^R}{J_r} \sum_{j = 1}^{J_r} \mathbbm{1}_{E^f_{K_r}[j] \leq v} - \frac{1}{2}(\mathbbm{1}_{E^c_{K_r}[j] \leq v} + \mathbbm{1}_{E^{c'}    _{K_r}[j] \leq v})$
\EndFor
\State \Return $\hat{I}$
\EndProcedure
\State $\hat{I} \gets$ \Call{Evaluate}{$u$}
\State $\hat{q} \gets$ a root of $v \mapsto$ (\Call{Evaluate}{$v$}$ - p$)
\State \Return $(\hat{I}, \hat{q})$
\end{algorithmic}
\label{alg:simulatenous_estim_ml2r}
\end{algorithm}

In Figure~\ref{fig:all_estimator_benchmark}, we present the computational cost of each estimator as a function of the achieved empirical RMSE for both target quantities. We observe that the performance of each estimator is consistent across the two tasks. The weighted Multi-level Monte Carlo estimator performs clearly better than the other estimators. We also notice that the new parameters from Table \ref{tab:num_optim_mlmc_parameters} increase the performance on lower computational budget while keeping the increased performance on the highest budgets. Contrary to the standard parameters of Table \ref{tab:closed_mlmc_parameters_v2}, the new parameters correctly select $R = 1$ when appropriate (here roughly for computational budgets lower than $10^7$). In other words the MLMC with optimized parameters are always at least as efficient as a standard Nested MC. Notice that, at the highest achieved precisions, the ML2R estimator with parameters calibrated using Table \ref{tab:num_optim_mlmc_parameters} (orange curve) is approximately 3 to 4 times more efficient than the traditional nested Monte Carlo (MC) method (blue curve). Additionally, observe that the MLMC method without weights (red curve) does not outperform the traditional nested MC in this experiment. This supports the theoretical view that the use of weights is crucial when dealing with indicator functions.

\begin{figure}[H]
    \centering
    \subfloat[\centering Pointwise c.d.f. estimation at 99.5\% level. ]{{\includegraphics[width=.45 \textwidth]{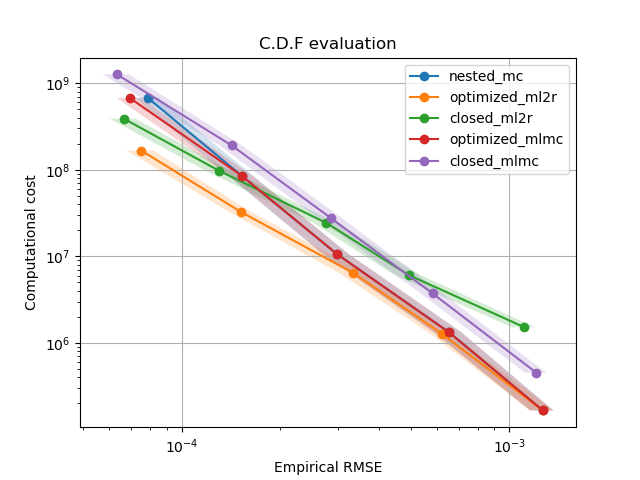} }}
    \subfloat[\centering 99.5\% Quantile estimation.]{{\includegraphics[width=.45 \textwidth]{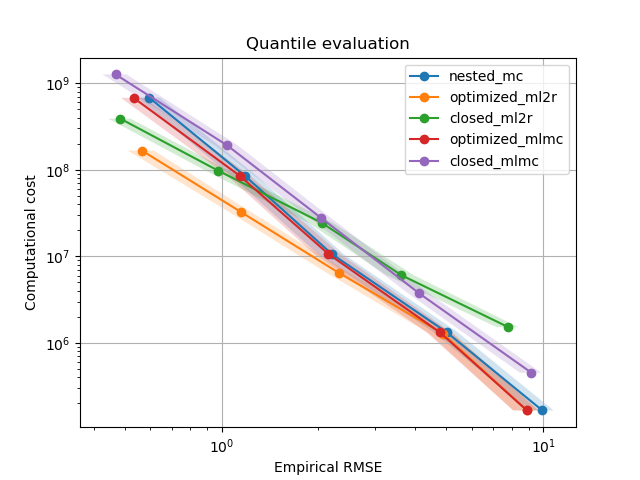} }}
    \vspace{0.5cm}
    \caption{Comparison of different computational complexity : Empirical RMSE against computational cost. In faded color the 95\% confidence interval on the RMSE. In blue, the traditional Nested MC estimator calibrated with Table \ref{tab:num_optim_mlmc_parameters} and R = 1. In orange the ML2R estimator calibrated with Table \ref{tab:num_optim_mlmc_parameters}. In green the ML2R estimator calibrated with Table \ref{tab:closed_mlmc_parameters_v2}. In red the standard MLMC estimator calibrated with Table \ref{tab:num_optim_mlmc_parameters}. In purple the standard MLMC estimator calibrated with Table \ref{tab:closed_mlmc_parameters_v2}.
    \label{fig:all_estimator_benchmark}}
\end{figure}

\subsection{Influence of $\tau$ on efficiency.}
\label{sec:tau_num_exp}
In this section, we compare the robustness with respect to the parameter $\tau$ of an ML2R estimator calibrated with Table \ref{tab:closed_mlmc_parameters_v2} to an ML2R calibrated with Table \ref{tab:num_optim_mlmc_parameters}. The methodology consist of considering a fixed computational budget $C \in \mathbb{N}$ and an increasing sequence $(\tau_{n})_{n \in \{1, \dots, N\}}$ of parameter $\tau$. For each $n \in \{1, \dots, N\}$, since both $\mathcal{C}_{\tau}(\theta_{\varepsilon})$ and $\mathcal{C}_{\tau}(\theta^*_{\varepsilon})$ are increasing when $\varepsilon$ decrease, we can find with a root finding algorithm, $\varepsilon(\tau)$ and $\varepsilon^*(\tau)$ such that $\mathcal{C}_{\tau}(\theta_{\varepsilon(\tau)}) \approx \mathcal{C}_{\tau}(\theta^*_{\varepsilon^*(\tau)}) \approx C$. These parameters correspond to the best use of the given computational budget for both calibration approaches. We then compute the efficiency $e_{\tau}$ of parameters $\theta^*_{\varepsilon^*(\tau)}$ with respect to $\theta(\varepsilon_{\tau})$ with $e_{\tau} := \frac{\mathcal{M}(\theta(\varepsilon_{\tau}))}{\mathcal{M}(\theta^*(\varepsilon^*_{\tau}))}$ where both MSE are estimated empirically. In Figure \ref{fig:tau_efficiency} we plot the efficiency for varying values of $\tau$ (see also Table \ref{tab:closed_ml2r_tau} and Table \ref{tab:optim_ml2r_tau}). We observe that the efficiency is increasing when $\tau$ grows, which confirm that parameters from Table \ref{tab:num_optim_mlmc_parameters} are more robust to large $\tau$. Furthermore observe in Table \ref{tab:optim_ml2r_tau} that as $\tau$ grows, the selected $R$ from Table \ref{tab:num_optim_mlmc_parameters} decreases, confirming the discussion of Section \ref{sec:tau_analysis}. However, since the selected $R$ remains greater than 1 even for larger values of $\tau$, we see that the weighted MLMC procedure continues to be more efficient than the traditional nested estimator.

\begin{figure}[H]
    \centering
    \includegraphics[width=0.7\linewidth]{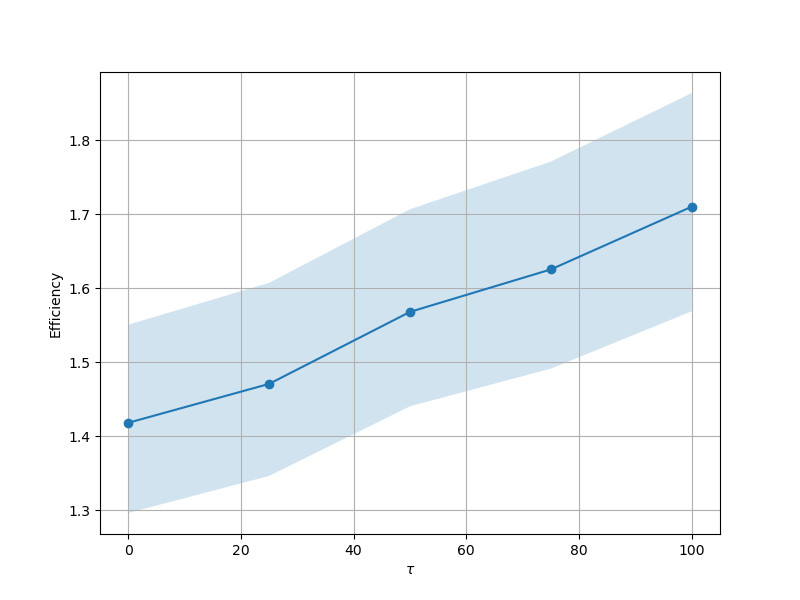}
    \caption{Efficiency of the optimized parameters of Table \ref{tab:num_optim_mlmc_parameters} against parameters of Table \ref{tab:closed_mlmc_parameters_v2} for an ML2R estimator. In faded color, the 90\% confidence interval for the efficiency.}
    \label{fig:tau_efficiency}
\end{figure}

\begingroup
\setlength{\tabcolsep}{10pt}
\renewcommand{\arraystretch}{1.5}
\begin{table}[H]
\centering
\begin{tabular}{ c c c c c c}
\hline
$\tau$ & $J$ & $K$ & $R$ & Cost & CDF RMSE \\
\hline
$0$ & $1.37 \cdot 10^{7}$ & $10$ & $4$ & $5.00 \cdot 10^{8}$ & $6.51 \cdot 10^{-5} \; (6.20 \cdot 10^{-5} - 6.80 \cdot 10^{-5})$ \\
$25$ & $8.98 \cdot 10^{6}$ & $10$ & $4$ & $5.00 \cdot 10^{8}$ & $8.24 \cdot 10^{-5} \; (7.88 \cdot 10^{-5} - 8.59 \cdot 10^{-5})$ \\
$50$ & $6.19 \cdot 10^{6}$ & $10$ & $4$ & $5.00 \cdot 10^{8}$ & $9.74 \cdot 10^{-5} \; (9.32 \cdot 10^{-5} - 1.01 \cdot 10^{-4})$ \\
$75$ & $4.73 \cdot 10^{6}$ & $10$ & $4$ & $5.00 \cdot 10^{8}$ & $1.13 \cdot 10^{-4} \; (1.08 \cdot 10^{-4} - 1.20 \cdot 10^{-4})$ \\
$100$ & $3.82 \cdot 10^{6}$ & $10$ & $4$ & $5.00 \cdot 10^{8}$ & $1.25 \cdot 10^{-4} \; (1.20 \cdot 10^{-4} - 1.30 \cdot 10^{-4})$ \\
\hline
\end{tabular}
\vspace{0.5cm}
\caption{Sensitivity of an ML2R calibrated with Table \ref{tab:closed_mlmc_parameters_v2} to varying $\tau$. In parenthesis are 95\% confidence intervals on the estimated RMSE. For conciseness parameter $q$ is omitted.}
\label{tab:closed_ml2r_tau}
\end{table}
\endgroup

\begingroup
\setlength{\tabcolsep}{10pt}
\renewcommand{\arraystretch}{1.5}
\begin{table}[H]
\centering
\begin{tabular}{ c c c c c c}
\hline
$\tau$ & $J$ & $K$ & $R$ & Cost & CDF RMSE\\
\hline
$0$ & $2.23 \cdot 10^{7}$ & $10$ & $3$ & $5.00 \cdot 10^{8}$ & $4.59 \cdot 10^{-5} \; (4.38 \cdot 10^{-5} - 4.78 \cdot 10^{-5})$ \\
$25$ & $6.30 \cdot 10^{6}$ & $38$ & $2$ & $5.00 \cdot 10^{8}$ & $5.60 \cdot 10^{-5} \; (5.34 \cdot 10^{-5} - 5.85 \cdot 10^{-5})$\\
$50$ & $4.71 \cdot 10^{6}$ & $39$ & $2$ & $5.00 \cdot 10^{8}$ & $6.21 \cdot 10^{-5} \; (5.95 \cdot 10^{-5} - 6.47 \cdot 10^{-5})$ \\
$75$ & $3.72 \cdot 10^{6}$ & $41$ & $2$ & $5.00 \cdot 10^{8}$ & $6.93 \cdot 10^{-5} \; (6.63 \cdot 10^{-5} - 7.21 \cdot 10^{-5})$ \\
$100$ & $3.08 \cdot 10^{6}$ & $43$ & $2$ & $5.00 \cdot 10^{8}$ & $7.30 \cdot 10^{-5} \; (6.98 \cdot 10^{-5} - 7.59 \cdot 10^{-5})$ \\
\hline
\end{tabular}
\vspace{0.5cm}
\caption{Sensitivity of an ML2R calibrated with Table \ref{tab:num_optim_mlmc_parameters} to varying $\tau$. In parenthesis are 95\% confidence intervals on the estimated RMSE. For conciseness parameter $q$ is omitted.}
\label{tab:optim_ml2r_tau}
\end{table}
\endgroup

\bibliographystyle{unsrt}  
\bibliography{bibliography}

\appendix
\section{Appendix}

\subsection{Computation of the weights in the ML2R}
\label{apx:details_weights}
In this section we follow the derivation of the weights in the ML2R of Pagès \cite{numProbaPages} Section 9.4 and 9.5.

For $\alpha > 0$, $R \in \mathbb{N}$ fixed, we consider $(w_1, \dots, w_R) \in \mathbb{R}^R$ the unique solution to the following Vandermonde system of equations :
\begin{equation}
    \begin{cases}
      \sum_{j = 1}^R w_j = 1 \\
      \sum_{j = 1}^R \frac{w_j}{(2^{\alpha(j-1)})^{r}}, r=1, \dots, R-1
    \end{cases} \,.
\end{equation}
Considering for $r = 2, \dots, R$, $K_r = K_1 2^{r - 1}$ and assuming that $(WE_{\alpha, R})$ holds
\begin{equation*}
    \mathbb{E} \left[ \sum_{r = 1}^R w_r Y_{K_r} \right] = \mathbb{E} \left[ Y_{\infty}\right] + \bar{W}_{R+1} \frac{c_R}{K_1^R} + o\left( \frac{1}{K_1^R}\right)
\end{equation*}
with $\bar{W}_{R+1} = \sum_{r = 1}^R \frac{w_r}{2^{\alpha r}}$. The weights effectively remove the bias up to a term of order $R$ and the solution $(w_1, \dots, w_R) \in \mathbb{R}^R$ has an explicit expression
\begin{equation*}
    \forall i \in \{1, \dots, R\}, w_i = \; \frac{(-1)^{R - i}}{\prod_{\substack{1 \leq j \leq R \\ j \neq i}} |1 - 2^{\alpha(j - i)}|} \,.
\end{equation*}
This gives (see for example \cite{lemairePages17}, Proposition 2.5)
\begin{equation}
    \mathbb{E} \left[ \sum_{r = 1}^R w_r Y_{K_r} \right] = \mathbb{E} \left[ Y_{\infty}\right] + (-1)^{R - 1} \frac{c_R}{K_1^{\alpha R} 2^{\frac{\alpha R(R-1)}{2}}} + o\left( \frac{1}{K_1^{\alpha R}}\right)
\end{equation}
Combining this Richardson-Romberg procedure with a Multi-level procedure gives the weights :
\begin{equation}
\label{eq:def_weights_ml2r}
    \forall 1 \leq i \leq R, \; W_{i} = w_i + w_{i+1} + \dots + w_R \,.
\end{equation}
Note that in particular $W_1 = 1$. Then using the Abel transform formula:
\begin{equation*}
    \mathbb{E} \left [ Y_{K_1} + \sum_{r = 2}^R W_r \left(Y_{K_r} - Y_{K_{r-1}}\right) \right] = \mathbb{E} \left[ \sum_{r = 1}^R w_r Y_{K_r} \right] \,.
\end{equation*}
This shows that if $A_r = W_r$, $r = 2, \dots, R$, then under $(WE_{\alpha, R})$, the weighted Multi-level procedure as no bias term up to the order $R$. To emphasize the dependence of the weights on $R$ we will refer to them as $(W_1^R, W_2^R, \dots, W_R^R)$ in this article. \\

\subsection{Technical Lemmas for Section \ref{sec:extension_complexity_thms}}
\label{apx:technical_lemmas_sec}

\begin{lemma}{(See \cite{Giorgi_2017}, Proposition 4.1)}
\label{lemma:K_underline}
    There exist a $\varepsilon_0 > 0$ such that for all $\varepsilon \leq \varepsilon_0$ such that $K(\varepsilon) = \underline{K}$.
\end{lemma}

\begin{lemma}[See \cite{Giorgi_2017}, Lemma 4.3]
\label{lemma:main_lemma_giorgi_weights}
    Let $\alpha > 0$, $R \in \mathbb{N}$ and $(W_j^R)_{j \in \{1, \dots, R\}}$ be the associated ML2R weights defined in (\ref{eq:def_weights_ml2r}).
    \begin{enumerate}
        \item the collection $\{ W_j^R : R \in \mathbb{N}, j \in \{1, \dots, R\}\}$ is uniformly bounded, that is there exists a constant $W_{\infty} > 0$ such that,
        \begin{equation*}
            \forall R \in \mathbb{N}, \; j \in \{1, \dots, R \}, \quad |W_{j}^R| \leq W_{\infty} \,,
        \end{equation*}

        \item for all $\gamma > 0$,
        \begin{equation*}
            \underset{R \to +\infty}{\lim} \sum_{j = 2}^R |W_j^R| 2^{-\gamma(j-1)} = \frac{1}{2^{\gamma} - 1} \,.
        \end{equation*}
    \end{enumerate}
\end{lemma}

\begin{lemma}[See \cite{Giorgi_2017}, Lemma 7.1]
\label{lemma:giorgi_lim_weights}
    Let $\phi : \mathbb{N} \to \mathbb{N}$ such that $\phi(R) \in \{1, \dots, R-1\}$ for every $R \geq 1$, $\phi(R) \underset{R \to +\infty}{\longrightarrow} +\infty$, and $R - \phi(R) \underset{R \to +\infty}{\longrightarrow} +\infty$. Then,
    \begin{equation*}
        \underset{R \to +\infty}{\lim} \underset{j \in \{1, \dots, \phi(R)\}}{\sup} |W_j^R - 1| = 0 \,.
    \end{equation*}
    In particular for all $j \in \mathbb{N}$, $W_j^R \underset{R \to +\infty}{\longrightarrow} 1$.
\end{lemma}

\begin{lemma}(See \cite{Giorgi_2017}, Section 4.4)
\label{lemma:mu_cv_q}
    Let $\varepsilon > 0$ and $\theta_{\varepsilon} \in \Theta$ the MLMC parameter sequence from Table \ref{tab:closed_mlmc_parameters_v2}. Consider, $\mu_{\varepsilon} = \sum_{r = 1}^{R_{\varepsilon}} \frac{\bar{\sigma}(r, K_{\varepsilon})}{\sqrt{\gamma_{\tau}(r, K_{\varepsilon})}}$ then, $\mu_{\varepsilon} \underset{\varepsilon \to 0}{\longrightarrow} \mu_{\infty} \in (0, +\infty)$.
\end{lemma}
\begin{proof}
    Recall that for all $\varepsilon > 0$, $K(\varepsilon)= \underline{K}$ , therefore since $\mu_{\varepsilon}$ is such that $\sum_{r = 1}^{R(\varepsilon)} q_r(\varepsilon) = 1$ we get,
    \begin{equation*}
        \mu_{\varepsilon} = \frac{\bar{\sigma}}{\sqrt{\underline{K}}} + \frac{\sqrt{V_1}}{\underline{K}^{\frac{1 + \beta}{2}}} \sum_{r = 2}^{R(\varepsilon)} \frac{|A_{r}^{R(\varepsilon)}|}{2^{^{\frac{(r-1)}{2} (1 + \beta)}}} \,.
    \end{equation*}
    Then since $\beta > 0$, an application of Lemma \ref{lemma:main_lemma_giorgi_weights} 2. with $\gamma = \frac{1+\beta}{2}$ yields the result.
\end{proof}

\subsection{Technical Lemmas for Section \ref{sec:complexity_analysis}}
\label{apx:technical_lemmas_for_complexity_thm}

\begin{lemma}
\label{lemma:proxy_bias_satisfied_for_K_R}
    Let $\varepsilon > 0$, $R \in \mathbb{N}$, define in the standard MLMC case $C_{\varepsilon, R} := \frac{\varepsilon}{\sqrt{1 + 2 \alpha}}$
    while in the ML2R case $C_{\varepsilon, R} := \frac{\varepsilon}{\sqrt{1 + 2 \alpha R}}$ then in both cases $|\tilde{\mu}(K(\varepsilon), R)| \leq C_{\varepsilon, R}$.
\end{lemma}
\begin{proof}
    We begin with the proof in the case of the standard MLMC estimator. Let $\varepsilon > 0$, starting from (\ref{eq:mu_tilde_mlmc}),
    \begin{equation*}
        |\tilde{\mu}(K(\varepsilon), R(\varepsilon)| \leq \frac{|c_1|}{K(\varepsilon)^{\alpha} 2^{(R(\alpha) -1) \alpha}} \leq \frac{|c_1|}{K^+(\varepsilon)^{\alpha} 2^{(R(\alpha) -1) \alpha}} = \frac{\varepsilon}{\sqrt{1 + 2\alpha}}
    \end{equation*}
    where in the last equality we used the definition of $K^+(\varepsilon)$ from Table \ref{tab:closed_mlmc_parameters_v2}. We continue with the proof in the case of the ML2R estimator. Let $\varepsilon > 0$, starting from (\ref{eq:proxy_bias_ml2r}),
    \begin{equation*}
        |\tilde{\mu}(K(\varepsilon), R(\varepsilon)| \leq \frac{\tilde{c}^R}{K(\varepsilon)^{\alpha R} 2^{\frac{R(\alpha)(R(\alpha) -1)}{2} \alpha}} \leq \frac{\tilde{c}^R}{K^+(\varepsilon)^{\alpha R} 2^{\frac{R(\alpha)(R(\alpha) -1)}{2} \alpha}} = \frac{\varepsilon}{\sqrt{1 + 2\alpha R(\varepsilon)}}
    \end{equation*}
    where in the last equality we used the definition of $K^+(\varepsilon)$ from Table \ref{tab:closed_mlmc_parameters_v2}. This conclude the proof.
\end{proof}

\begin{lemma}
    For all $\varepsilon > 0$,
    \begin{equation}
    \label{eq:J_eps_form}
        J(\varepsilon) = \frac{\bar{v}(\pi(\varepsilon))}{\varepsilon^2 - C^2_{\varepsilon, R(\varepsilon)}}
    \end{equation}
    where $C_{\varepsilon, R(\varepsilon)}$ is defined in Lemma \ref{lemma:proxy_bias_satisfied_for_K_R}.
\end{lemma}
\begin{proof}
    Let $\varepsilon > 0$ and denote $\pi(\varepsilon) = (q(\varepsilon), K(\varepsilon), R(\varepsilon))$. Recall from (\ref{eq:def_v_bar}) that
    \begin{equation*}
        \bar{v}(\pi(\varepsilon)) = \sum_{r = 1}^{R(\varepsilon)} \frac{\bar{\sigma}^2(r, K(\varepsilon))}{q_r(\varepsilon)} = \frac{\bar{\sigma}^2}{q_1(\varepsilon)} + \sum_{r = 2}^{R(\varepsilon)} \frac{V_1 (A^{R(\varepsilon)}_{r})^2}{K(\varepsilon)^{\beta} 2^{\beta(r-1)} q_r(\varepsilon)} \,.
    \end{equation*}
    Now since,
    \begin{equation*}
        \frac{1}{\varepsilon^2 - C^2_{\varepsilon, R(\varepsilon)}} = \frac{\varepsilon^2}{M_{\varepsilon}}
    \end{equation*}
    where $M_{\varepsilon}$ is defined in Table \ref{tab:closed_mlmc_parameters_v2}, we clearly get that
    \begin{equation*}
        \frac{\bar{v}(\pi(\varepsilon))}{\varepsilon^2 - C^2_{\varepsilon, R(\varepsilon)}} = J(\varepsilon)
    \end{equation*}
    proving the claim.
\end{proof}

\begin{lemma}
\label{lemma:m_tilde_satisfied_for_theta}
    For all $\varepsilon > 0$, $\widetilde{\mathcal{M}}(\theta_{\varepsilon}) \leq \varepsilon^2$.
\end{lemma}
\begin{proof}
    Let $\varepsilon > 0$, denote $\pi(\varepsilon) = (q(\varepsilon), K(\varepsilon), R(\varepsilon))$ and $\pi_0(\varepsilon) = (K(\varepsilon), R(\varepsilon))$. Recall from (\ref{eq:def_tilde_Mse}) that
    \begin{equation*}
        \widetilde{\mathcal{M}}(\theta_{\varepsilon}) = \frac{\bar{v}(\pi(\varepsilon))}{J(\varepsilon)} + \tilde{\mu}^2(\pi_0(\varepsilon))
    \end{equation*}
    Owing to (\ref{eq:J_eps_form}) we have,
    \begin{equation*}
        J(\varepsilon) = \frac{\bar{v}(\pi(\varepsilon))}{\varepsilon^2 - C^2_{\varepsilon, R(\varepsilon)}}
    \end{equation*}
    Therefore, $\widetilde{\mathcal{M}}(\theta_{\varepsilon}) = \varepsilon^2 - C^2_{\varepsilon, R(\varepsilon)} + \tilde{\mu}^2(\pi_0(\varepsilon))$. By Lemma (\ref{lemma:proxy_bias_satisfied_for_K_R}), $C_{\varepsilon, R(\varepsilon)} \geq |\tilde{\mu}(\pi_0(\varepsilon))|$ from which we get the result $\widetilde{\mathcal{M}}(\theta_{\varepsilon}) \leq \varepsilon^2$.
\end{proof}

\subsection{Proof of Proposition \ref{prop:nested_explicit_K}}
\label{apx:proof_nested_explicit_K}
To prove the proposition we introduce, for all $b \in \mathbb{R}$,
\begin{equation*}
	E_{\varepsilon, b} = \left \{ x \in (0, +\infty) : \frac{b}{\varepsilon} < x \right \}
\end{equation*}
and for all $a \in \mathbb{R}^+$, define the auxiliary function $f_{\varepsilon, a, b}$ from $E_{\varepsilon, b}$ to $\mathbb{R}$ such that
\begin{equation*}
	\forall x \in E_{\varepsilon}, \quad f_{\varepsilon, a, b}(x) = \frac{x + a}{\varepsilon^2 - \frac{b^2}{x^2}} = \frac{x^2 \left(x + a \right)}{x^2\varepsilon^2 - b^2}
\end{equation*}
\begin{lemma}
\label{lem:f_properties}
$f$ is a differentiable, strictly convex and coercive function on $x \in E_{\varepsilon}$, therefore it admits a unique global minima on $E_{\varepsilon}$.
\end{lemma}
\begin{proof}
    For all $\varepsilon > 0$, $a \in \mathbb{R}_+$, $b \in \mathbb{R}$, $f_{\varepsilon, a, b}$ is two time differentiable on $E_{\varepsilon}$ and for all $x \in E_{\varepsilon}$
\begin{equation*}
	f'_{\varepsilon, a, b}(x) = \frac{x(-2 a b^2 +\varepsilon^2 x^3 - 3b^2 x )}{(x^2 \varepsilon^2 - b^2)^2}, \quad f''_{\varepsilon, a, b}(x) = \frac{2 b^2}{(x^2 \varepsilon^2 - b^2)^3} \left( a \left(3 x^2 \varepsilon^2 + b^2 \right) + x \left(x^2 \varepsilon^2 + 3b^2\right) \right)
\end{equation*}
Therefore we have that $f''(x) > 0$, thus $f$ is strictly convex on $E_{\varepsilon}$. Furthermore it is easy to check that it is a coercive function.
\end{proof}

\begin{lemma}
\label{lem:optim_function}
For all $\varepsilon > 0$, $b \in \mathbb{R}$, if $a = 0$
\begin{equation*}
f'_{\varepsilon, 0, b}(x) = 0 \Leftrightarrow x = \frac{\sqrt{3}b}{\varepsilon}
\end{equation*}
if $a > 0$,
\begin{equation*}
    f'_{\varepsilon, a, b}(x) \Leftrightarrow \begin{cases}
	\frac{|b|}{\varepsilon} 2\cos \left( \frac{1}{3} \arccos \left( \frac{\varepsilon a}{|b|} \right)\right) &, \; \varepsilon < \frac{|b|}{a} \\
	3 b &, \; \varepsilon = \frac{|b|}{a} \\
	\left( \frac{|b|}{\varepsilon} \right)^{\frac{2}{3}} \left[ \left( a + \left( a^2 - \frac{b^2}{\varepsilon^2} \right)^{\frac{1}{2}} \right)^{\frac{1}{3}} + \left( a - \left(  a^2 - \frac{b^2}{\varepsilon^2} \right)^{\frac{1}{2}} \right)^{\frac{1}{3}}\right]&, \; \varepsilon > \frac{|b|}{a} \,.
	\end{cases}
\end{equation*}
\end{lemma}
\begin{proof}
From Lemma \ref{lem:f_properties}, for all $\varepsilon > 0$, $a \in \mathbb{R_+}$ and $b \in \mathbb{R}$, $f_{\varepsilon, a, b}$ admits a unique minimum $x_*$ in $E_{\varepsilon, b}$ characterized by 
\begin{equation*}
f_{\varepsilon, a, b}'(x_*) = 0
\end{equation*}

We have that for all $x \in E_{\varepsilon, b}$
\begin{align*}
f_{\varepsilon, a, b}'(x) = \frac{x}{(x^2 \varepsilon^2 - c_{1}^2)} \left(-2 \tau c_{1}^2 + \varepsilon^2 x^3 - 3c_1^2 x \right)
\end{align*}
Therefore,
\begin{equation*}
	f_{\varepsilon, a, b}'(x) = 0 \Longleftrightarrow x^3 - \frac{3 b^2}{\varepsilon^2}x - \frac{2 b^2 a}{\varepsilon^2} = 0
\end{equation*}
If $a = 0$, then we easily obtain that
\begin{equation*}
    f_{\varepsilon, a, b}'(x) = 0 \Longleftrightarrow x = \frac{\sqrt{3} b}{\varepsilon}
\end{equation*}
Otherwise when $\tau > 0$ we consider the polynomial
\begin{equation*}
	P[X] = X^3 - \frac{3 b^2}{\varepsilon^2} X - \frac{2 b^2 a }{\varepsilon^2}
\end{equation*}
and its discriminant :
\begin{equation*}
	\Delta = - \left(4 \left( - \frac{3 b^2}{\varepsilon^2} \right)^3 + 27 \left(\frac{2 b^2 a}{\varepsilon^2} \right)^2 \right) = 108\frac{b^4}{\varepsilon^4}\left(\frac{b^2}{\varepsilon^2} - a^{2}\right)
\end{equation*}
Therefore,
\begin{equation*}
	\begin{cases}
		\Delta > 0 &\Longleftrightarrow \varepsilon < \frac{|b|}{a} \\
		\Delta = 0 &\Longleftrightarrow \varepsilon = \frac{|b|}{a} \\
		\Delta < 0 &\Longleftrightarrow \varepsilon > \frac{|b|}{a}
	\end{cases}
\end{equation*}
To study the root of $P[X]$ we study the three cases separately. \\

\textbf{Case 1:} We assume that $\varepsilon < \frac{|b|}{a}$, therefore $\Delta > 0$, then from Cardano's cubic formulas, $P[X]$ admits three distinct real roots $z_0, z_1, z_2$ such that for $k \in \{0, 1, 2 \}$ :
\begin{equation*}
	z_k = 2 \frac{|b|}{\varepsilon} \cos \left( \frac{1}{3} \arccos \left( \frac{\varepsilon a}{|b|} \right) + \frac{2k \pi}{3}\right)
\end{equation*}
Since there exist at least one $x$ in $E_{\varepsilon}$ such that $f_{\varepsilon, a, b}'(x) = 0$, we know that one of the root is located in $E_{\varepsilon}$ (i.e is such that $z_k > \frac{|b|}{\varepsilon}$). Remark that
\begin{equation*}
	z_k > \frac{|b|}{\varepsilon} \Longleftrightarrow \cos \left( \frac{1}{3} \arccos \left( \frac{\varepsilon a}{|b|} \right) + \frac{2k \pi}{3}\right) > \frac{1}{2}
\end{equation*}
From the fact that $\cos$ is decreasing on $[0, \pi]$, $2\pi$-periodic even and that $\cos(\frac{\pi}{3}) = \frac{1}{2}$ we get that for all $x \in \mathbb{R}$
\begin{equation}
\label{eq:cos_sol_half}
	\cos(x) > \frac{1}{2} \Longleftrightarrow \frac{-\pi}{3} < x < \frac{\pi}{3} \quad [2\pi]
\end{equation}
Now since $\arccos(]0, 1[) \subset ]0, \pi[$ we have that :
\begin{equation*}
	\frac{2k \pi}{3} < \frac{1}{3} \arccos \left( \frac{\varepsilon a}{|b|} \right) + \frac{2k \pi}{3} < \pi + \frac{2k\pi}{3}
\end{equation*}
Among the three intervals $]0, \pi[$, $\left ] \frac{2\pi}{3}, \frac{5\pi}{3} \right[$, $\left ]\frac{4\pi}{3}, \frac{7\pi}{3} \right[$ only the interval $]0, \pi[$ contains values that that satisfies Conditions \ref{eq:cos_sol_half}. Therefore we have necessarily that $z_0 \in E_{\varepsilon}$ and $f_{\varepsilon, a, b}'(z_0) = 0$. \\

\textbf{Case 2:} Now we assume that $\varepsilon = \frac{|b|}{a}$, therefore $\Delta = 0$, in this case from Cardano's formulas $P[X]$ admits two real roots $(z_0, z_1) \in \mathbb{R}^2$
\begin{equation*}
	\begin{cases}
		z_0 &= 3 b \\
		z_1 &= -3 b
	\end{cases}
\end{equation*}
Since $z_1 < 0$ we have $z_1 \notin E_{\varepsilon}$, therefore necessarily $z_0 \in E_{\varepsilon}$ and $f_{\varepsilon, a, b}'(z_0) = 0$. \\

\textbf{Case 3:} Now we assume that $\varepsilon > \frac{|b|}{a}$, thus $\Delta < 0$. From Cardano's formulas $P[X]$ admits one real root $z_0 \in \mathbb{R}^2$
\begin{equation*}
	z_0 = \left( \frac{b}{\varepsilon} \right)^{\frac{2}{3}} \left[ \left( \tau + \left( a^2 - \frac{b^2}{\varepsilon^2} \right)^{\frac{1}{2}} \right)^{\frac{1}{3}} + \left( a - \left(  a^2 - \frac{b^2}{\varepsilon^2} \right)^{\frac{1}{2}} \right)^{\frac{1}{3}}\right]
\end{equation*}

Then a direct application of Lemma \ref{lem:optim_function} with $a = \tau$ and $b = c_1$ yields the result.
\end{proof}

\subsection{Technical Lemma for Proposition \ref{prop:closed_formula_neg_OF}}
\label{apx:technical_lemmas_for_alm}
In this appendix we present some technical Lemmas used to prove Proposition \ref{prop:closed_formula_neg_OF} in Section \ref{sec:num_appli}.

\label{apx:technical_lemma_of}
\begin{definition}
    We define $\rho$, the function form $(0, +\infty)^2$ to $\mathbb{R}$ defined by,
    \begin{equation*}
        \forall (x, x') \in (0, +\infty)^{2}, \quad \rho(x, x') = \max\left(r_g, \gamma\ln\left(\frac{x'}{x} \right) \right)
    \end{equation*}
    In particular we have,
    \begin{equation}
    \label{eq:served_rate_state_function}
        \forall t \in \{1, \dots, T\}, \quad r_s(t) = \rho(S_{t-1},  S_t)
    \end{equation}
\end{definition}

\begin{lemma}
\label{lemma:mr_rec_formula_alm}
    The following holds,
    \begin{equation}
    \label{eq:rec_MR}
         \forall t \in \{0, \dots, T \}, \forall u \in \{0, \dots, t\}, \quad MR_t = MR_u \prod_{i = u+1}^{t} (1 - d_i)(1 + \rho(S_{i-1}, S_{i}))
    \end{equation}

    For all $t \in \{0, \dots, T \}$, let $f_t$ be the function from $(0, +\infty)^{t+1}$ to $\mathbb{R}$ defined by,
    \begin{equation*}
        \forall x = (x_0, \dots, x_t) \in (0, +\infty)^{t+1}, \quad f_t(x) = MR_0 \prod_{u = 1}^{t} (1 - d_u) (1 + \rho(x_{u-1}, x_u))
    \end{equation*}
    then,
    \begin{equation}
    \label{eq:proof_lemma_mr_function}
        MR_t = f_t(S_0, S_1, \dots, S_t)\,.
    \end{equation}
    In particular $MR_t$ is $\mathcal{F}_t$-measurable.
\end{lemma}
\begin{proof}
    A straightforward induction argument yields the first result. Let $t \in \{0, \dots, T\}$, clearly (\ref{eq:rec_MR}) holds for $u = t$. Let $u \in \{1, \dots, t\}$ be such that (\ref{eq:rec_MR}) holds, then by (\ref{eq:alm_update_mr}) and (\ref{eq:alm_update_mr_tilde}),
    \begin{equation*}
        MR_{u} = \widetilde{MR}_{u-1} (1 - d_{u}) = MR_{u-1} (1 + \rho(S_{u-1}, S_{u}))(1-d_{u}) \,.
    \end{equation*}
    Then using the induction hypothesis,
    \begin{equation*}
        MR_t = MR_u \prod_{i = u+1}^{t} (1 - d_i)(1 + \rho(S_{i-1}, S_i)) = MR_{u-1} \prod_{i = u}^{t} (1 - d_i)(1 + \rho(S_{i-1}, S_i))
    \end{equation*}
    proving the fist claim. The second claim is a straightforward application of (\ref{eq:rec_MR}) for $u = 0$.
\end{proof}

\begin{lemma}
\label{lemma:rec_relation_phi_alm}
    For all $t \in \{0, \dots, T \}$, and all $u \in \{0, \dots, t\}$
    \begin{equation}
    \label{eq:rec_phi}
        \phi_t = \phi_u - MR_u \sum_{i = u+1}^{t} \frac{d_i}{S_i} \prod_{j = u + 1}^{i-1} (1 - d_j) \prod_{j = u + 1}^i (1 + \rho(S_{j-1}, S_j)) \,.
    \end{equation}
    Let $g_t$ be the function from $(0, +\infty)^{t+1}$ to $\mathbb{R}$ defined by,
    \begin{equation*}
        \forall x = (x_0, x_1, \dots, x_t) \in (0, +\infty)^{t+1}, \quad g_t(x) = \phi_0 - MR_0 \sum_{i = 1}^{t} \frac{d_i}{x_i} \prod_{j = u + 1}^{i-1} (1 - d_j) \prod_{j = u + 1}^i (1 + \rho(x_{j-1}, x_j))
    \end{equation*}
    then,
    \begin{equation}
    \label{eq:phi_eq_g}
        \phi_t = g_t(S_0, S_1, \dots, S_t) \,.
    \end{equation}
    In particular $\phi_t$ is $\mathcal{F}_t$-measurable.
\end{lemma}
\begin{proof}
    Let $t \in \{0, \dots, T\}$, we begin by showing that for all $u \in \{0, \dots, t\}$
    \begin{equation}
    \label{eq:proof_lemma_phi}
        \phi_t = \phi_u - \sum_{i = u+1}^{t} \frac{\widetilde{MR}_i}{S_i}
    \end{equation}
    A straightforward induction argument prove this claim : clearly (\ref{eq:proof_lemma_phi}) holds for $u = t$. Let $u \in \{1, \dots, t\}$ be such that (\ref{eq:proof_lemma_phi}) holds. Since $u \in \{1, \dots, t\}$,
    \begin{equation*}
        \phi_u = \phi_{u-1} + \Delta\phi_u = \phi_{u-1} - \frac{\widetilde{MR}_u d_u}{S_u}
    \end{equation*}
    Therefore from the induction hypothesis we get,
    \begin{equation*}
        \phi_t = \phi_u - \sum_{i = u+1}^{t} \frac{\widetilde{MR}_i}{S_i} = \phi_{u-1} - \sum_{i=u}^{t} \frac{\widetilde{MR}_i}{S_i}
    \end{equation*}
    proving the claim. Now, for all $i \in \{u+1, \dots, t \}$,
    \begin{equation*}
        \widetilde{MR}_i = MR_{i-1} ( 1 + \rho(S_{i-1}, S_i))
    \end{equation*}
    Using (\ref{eq:rec_MR}),
    \begin{equation*}
        MR_{i-1} = MR_u \prod_{j = u+1}^{i-1} (1 - d_j) (1 + \rho(S_{j-1}, S_j)) \,.
    \end{equation*}
    Therefore,
    \begin{equation}
    \label{eq:proof_rec_phi_mr}
        \widetilde{MR}_i = MR_u \prod_{j = u+1}^{i-1} (1 - d_j) \prod_{j = u+1}^{i} (1 + \rho(S_{j-1}, S_j)) \,.
    \end{equation}
    Plugging (\ref{eq:proof_rec_phi_mr}) in (\ref{eq:proof_lemma_phi}) then gives the first claim. The second claim is a straightforward application of the first claim for $u = 0$.
\end{proof}

\begin{lemma}
\label{lemma:z_formula_alm}
    Let $t \in \{1, \dots, T\}$, we let,
    \begin{equation}
    \label{eq:def_z}
        z := \mathbb{E}_{\mathbb{Q}}[1 + r_s(t)] = \mathbb{E}_{\mathbb{Q}}[1 + \rho(S_{t-1}, S_t)]
    \end{equation}
    and
    \begin{equation*}
       d := \frac{r - \frac{\sigma^2}{2} - \frac{r_g}{\gamma}}{\sigma} 
    \end{equation*}
    then,
    \begin{equation*}
        z = 1 + r_g + \gamma \sigma (\Phi^{'}(d) + d \Phi(d))
    \end{equation*}
    where $\Phi$ is the c.d.f. of a standard normal random variable.
\end{lemma}
\begin{proof}
    Let $t \in \{1, \dots, T\}$, by definition,
    \begin{equation*}
        \mathbb{E}_{\mathbb{Q}}[r_s(t)] = r_g + \mathbb{E} \left[ \left(\gamma \ln\left(\frac{S_t}{S_{t-1}} \right) - r_g \right)_+ \right]
    \end{equation*}
    Now, standard Black-Scholes like calculations gives,
    \begin{equation*}
        \mathbb{E}_{\mathbb{Q}} \left[\left(\gamma \ln\left(\frac{S_t}{S_{t-1}}\right) - r_g \right)_+ \right] = \gamma \mathbb{E}_{\mathbb{Q}}\left[\ln\left(\frac{S_t}{S_{t-1}}\right) \mathbbm{1}_{ \gamma \ln\left(\frac{S_t}{S_{t-1}}\right) \geq r_g}\right] - r_g \mathbb{Q} \left( \gamma \ln\left(\frac{S_t}{S_{t-1}}\right) \geq r_g \right)
    \end{equation*}
    Letting $N$ be a standard normal random variable, then,
    \begin{equation*}
        \ln\left(\frac{S_t}{S_{t-1}}\right) \overset{\mathrm{Law}}{=} r-\frac{\sigma^2}{2} + \sigma N
    \end{equation*}
    Therefore,
    \begin{equation*}
        \mathbb{Q}\left( \gamma \ln\left(\frac{S_t}{S_{t-1}}\right) \geq r_g\right) = \mathbb{Q}\left(N \geq -d\right) = \Phi(d)
    \end{equation*}
    and,
    \begin{equation*}
        \mathbb{E}_{\mathbb{Q}} \left[\ln\left(\frac{S_t}{S_{t-1}}\right) \mathbbm{1}_{ \gamma \ln\left(\frac{S_t}{S_{t-1}}\right) \geq r_g} \right] = \left(r - \frac{\sigma^2}{2}\right) \Phi(d) + \sigma \mathbb{E}_{\mathbb{Q}}[N \mathbbm{1}_{N \geq -d}] = \left(r - \frac{\sigma^2}{2}\right) \Phi(d) + \sigma \Phi'(d)
    \end{equation*}
    Which yields the result.
\end{proof}
\end{document}